\documentclass{article}

\usepackage[nonatbib, preprint]{neurips_2021}
\usepackage{bbm,amsmath}
\usepackage[numbers]{natbib}
\usepackage[ruled,vlined]{algorithm2e}
\usepackage{graphicx}
\usepackage{float}

\usepackage[utf8]{inputenc} 
\usepackage[T1]{fontenc}    
\usepackage{hyperref}       
\usepackage{url}            
\usepackage{booktabs}       
\usepackage{amsfonts}       
\usepackage{nicefrac}       
\usepackage{microtype}      
\usepackage{xcolor}         

\title{Adaptive Conformal Inference\\
  Under Distribution Shift}

\author{%
  Isaac Gibbs\\
  Department of Statistics\\
  Stanford University\\
  \texttt{igibbs@stanford.edu} \\
   \And
   Emmanuel J. Cand\`{e}s \\
   Department of Statistics\\
   Department of Mathematics \\
   Stanford University\\
   \texttt{candes@stanford.edu} \\
}

\newcommand{\mme}[0]{\mathbb{E}}
\newcommand{\mmp}[0]{\mathbb{P}}
\newcommand{\mmr}[0]{\mathbb{R}}
\newcommand{\mmn}[0]{\mathbb{N}}

\newcommand{\mmz}[0]{\mathbb{Z}}
\newcommand{\bone}[0]{\mathbbm{1}}

\newtheorem{theorem}{Theorem}[section]
\newtheorem{definition}{Definition}[section]
\newtheorem{proposition}{Proposition}[section]

\newtheorem{lemma}{Lemma}[section]

\newtheorem{example}{Example}[section]

\newenvironment{proof}{\paragraph{Proof:}}{\hfill$\square$}

\begin{document}

\maketitle

\begin{abstract}
  We develop methods for forming prediction sets in an online setting
  where the data generating distribution is allowed to vary over time
  in an unknown fashion. Our framework builds on ideas from conformal
  inference to provide a general wrapper that can be combined with any
  black box method that produces point predictions of the unseen label
  or estimated quantiles of its distribution. While previous conformal
  inference methods rely on the assumption that the data points are
  exchangeable, our adaptive approach provably achieves the desired
  coverage frequency over long-time intervals irrespective of the true data
  generating process. We accomplish this by modelling the distribution
  shift as a learning problem in a single parameter whose optimal
  value is varying over time and must be continuously re-estimated. We
  test our method, {\em adaptive conformal inference}, on two real
  world datasets and find that its predictions are robust to visible
  and significant distribution shifts.
\end{abstract}

\section{Introduction}

Machine learning algorithms are increasingly being employed in high stakes decision making processes. For instance, deep neural networks are currently being used in self-driving cars to detect nearby objects \cite{Badue2021} and parole decisions are being made with the assistance of complex models that combine over a hundred features \cite{Angwin2016}. As the popularity of black box methods and the cost of making wrong decisions grow it is crucial that we develop tools to quantify the uncertainty of their predictions.  

In this paper we develop methods for constructing prediction sets that are guaranteed  to contain the target label with high probability. We focus specifically on an online learning setting in which we observe covariate-response pairs $\{(X_t,Y_t)\}_{t \in \mmn} \subseteq \mathbb{R}^d \times \mathbb{R}$ in a sequential fashion. At each time step $t\in \mmn$ we are tasked with using the previously observed data $\{(X_r,Y_r)\}_{1 \leq r \leq t-1}$ along with the new covariates, $X_{t}$, to form a prediction set $\hat{C}_{t}$ for $Y_{t}$. Then, given a target coverage level $\alpha \in (0,1)$ our generic goal is to guarantee that $Y_{t}$ belongs to $\hat{C}_{t}$ at least $100(1-\alpha)$\% of the time.

Perhaps the most powerful and flexible tools for solving this problem come from conformal inference \cite[see e.g.][]{VovkBook, Gammerman2007, Shafer2008, Lei2014, Sadinle2019, Barber2020, Baber2021} . This framework provides a generic methodology for transforming the outputs of any black box prediction algorithm into a prediction set. The generality of this approach has facilitated the development of a large suite of conformal methods, each specialized to a specific prediction problem of interest \cite[e.g.][]{Romano2020b, Cauchois2021, Lei2018, Candes2021, Lei2021, Kivaranovic2020}. With only minor exceptions all of these algorithms share the same common guarantee that if the training and test data are exchangeable, then the prediction set has valid marginal coverage $\mmp(Y_{t} \in \hat{C}_{t}) = 1-\alpha$. 

While exchangeability is a common assumption, there are many real-world applications in which we do not expect the marginal distribution of $(X_t,Y_t)$ to be stationary. For example, in finance and economics market behaviour can shift drastically in response to new legislation or major world events. Alternatively, the distribution of $(X_t,Y_t)$ may change as we deploy our prediction model in new environments. This paper develops \textit{adaptive conformal inference} (ACI), a method for forming prediction sets that are robust to changes in the marginal distribution of the data. Our approach is both simple, in that it requires only the tracking of a single parameter that models the shift, and general as it can be combined with any modern machine learning algorithm that produces point predictions or estimated quantiles for the response. We show that over long time intervals ACI achieves the target coverage frequency without any assumptions on the data-generating distribution. Moreover, when the distribution shift is small and the prediction algorithm takes a certain simple form we show that  ACI will additionally obtain approximate marginal coverage at most time steps.

\subsection{Conformal inference} \label{sec:conformal_inference}

Suppose we are given a fitted regression model for predicting the value of $Y$ from $X$. Let $y$ be a candidate value for $Y_{t}$. To determine if $y$ is a reasonable estimate of $Y_{t}$, we define a conformity score $S(X,Y)$ that measures how well the value $y$ \emph{conforms} with the predictions of our fitted model. For example, if our regression model produces point predictions $\hat{\mu}(X)$ then we could use a conformity score that measures the distance between $\hat{\mu}(X_{t})$ and $y$. One such example is 
\[
S(X_{t},y) = |\hat{\mu}(X_{t}) - y|.
\]
Alternatively, suppose our regression model outputs estimates $\hat{q}(X;p)$ of the $p$th quantile of the distribution of $Y|X$. Then, we could use the method of conformal quantile regression (CQR) \cite{Romano2019}, which examines the signed distance between $y$ and fitted upper and lower quantiles through the score
\[
S(X_{t},y) = \max\{\hat{q}(X_{t};\alpha/2) -y, y-\hat{q}(X_t;1-\alpha/2) \}.
\]
Regardless of what conformity score is chosen the key issue is to
determine how small $S(X_{t},y)$ should be in order to accept $y$ as a
reasonable prediction for $Y_{t}$. Assume we have a calibration set
$\mathcal{D}_{\text{cal}} \subseteq \{(X_r,Y_r)\}_{1 \leq r \leq t-1}$ that is different from the data that was
used to fit the regression model. Using this calibration set we define
the fitted quantiles of the conformity scores to be
\begin{equation}\label{eq:split_quantile}
\hat{Q}(p) := \inf\left\{ s : \left( \frac{1}{|\mathcal{D}_{\text{cal}}|} \sum_{(X_r,Y_r) \in \mathcal{D}_{\text{cal}}} \bone_{\{S(X_r,Y_r) \leq s\}} \right) \geq p \right\}, 
\end{equation}
and say that $y$ is a reasonable prediction for $Y_{t}$ if $S(X_{t},y) \leq  \hat{Q}(1-\alpha)$. 

The crucial observation is that if the data
$\mathcal{D}_{\text{cal}} \cup \{(X_{t},Y_{t})\}$ are
exchangeable and we break ties uniformly at random then the rank of
$S(X_{t},Y_{t})$ amongst the points
$\{S(X_r,Y_r)\}_{(X_r,Y_r) \in \mathcal{D}_{\text{cal}}} \cup \{S(X_{t},Y_{t})\}$ will be
uniform. Therefore,
\[
\mmp(S(X_{t},Y_{t}) \leq \hat{Q}(1-\alpha)) =  \frac{\lceil |\mathcal{D}_{\text{cal}}|(1-\alpha) \rceil }{|\mathcal{D}_{\text{cal} }| + 1}.
\]
Thus, defining our prediction set to be
$\hat{C}_{t} := \{y : S(X_{t},y) \leq \hat{Q}(1-\alpha)\}$ gives the marginal coverage guarantee
\[
\mmp(Y_{t} \in \hat{C}_{t}) = \mmp(S(X_{t},Y_{t}) \leq \hat{Q}(1-\alpha)) =  \frac{\lceil |\mathcal{D}_{\text{cal}}|(1-\alpha) \rceil }{|\mathcal{D}_{\text{cal}}| + 1}.
\]
By introducing additional randomization this generic procedure can be altered slightly to produce a set $\hat{C}_{t}$ that satisfies the exact marginal coverage guarantee $\mmp(Y_{t} \in \hat{C}_{t}) = 1-\alpha$ \cite{VovkBook}. For the purposes of this paper this adjustment is not critical and so we omit the details here. Additionally, we remark that the method outlined above is often referred to as \textit{split} or \textit{inductive} conformal inference  \cite{Papadopoulos2002, VovkBook, Papadopoulos2008SplitConf}. This refers to the fact that we have split the observed data between a training set used to fit the regression model and a withheld calibration set. The adaptive conformal inference method developed in this article can also be easily adjusted to work with full conformal inference in which data splitting is avoided at the cost of greater computational resources \cite{VovkBook}.

\section{Adapting conformal inference to distribution shifts}\label{sec:main_method}

Up until this point we have been working with a single score function $S(\cdot)$ and quantile function $\hat{Q}(\cdot)$. In the general case where the distribution of the data is shifting over time both these functions should be regularly re-estimated to align with the most recent observations. Therefore, we assume that at each time $t$ we are given a fitted score function $S_{t}(\cdot)$ and corresponding quantile function $\hat{Q}_{t}(\cdot)$. We define the realized miscoverage rate of the prediction set $\hat{C}_{t}(\alpha) :=  \{y : S_t(X_{t},y) \leq  \hat{Q}_t(1-\alpha)\}$ as
\[
M_t(\alpha) := \mmp(S_t(X_{t},Y_{t}) > \hat{Q}_{t}(1-\alpha)),
\]
where the probability is over the test point $(X_{t},Y_{t})$ as well as the data used to fit $S_{t}(\cdot)$ and $\hat{Q}_{t}(\cdot)$. 

Now, since the distribution generating the data is non-stationary we do not expect $M_{t}(\alpha)$ to be equal, or even close to, $\alpha$. Even so, we can still postulate that if the conformity scores used to fit $\hat{Q}_t(\cdot)$ cover the bulk of the distribution of $S_t(X_{t},Y_{t})$ then there may be an alternative value $\alpha^*_t \in [0,1]$ such that $M_t(\alpha^*_t ) \cong \alpha$. More rigorously, assume that with probability one, $\hat{Q}_t(\cdot)$ is continuous, non-decreasing and such that $\hat{Q}_t(0) = -\infty$ and $\hat{Q}_t(1) = \infty$. This does not hold for the split conformal quantile functions defined in (\ref{eq:split_quantile}), but in the case where there are no ties amongst the conformity scores we can adjust our definition to guarantee this by smoothing over the jump discontinuities in $\hat{Q}(\cdot)$. Then, $M_t(\cdot)$ will be non-decreasing on $[0,1]$ with $M_t(0) = 0$ and $M_t(1) = 1$ and so we may define 
\[
\alpha^*_t := \sup\{\beta \in [0,1]: M_t(\beta) \leq \alpha\}.
\]
Moreover, if we additionally assume that 
\[
\mmp(S_t(X_{t},Y_{t}) = \hat{Q}_{t}(1-\alpha^*_t)) = 0, 
\]
then we will have that $M_t(\alpha^*_t) = \alpha$. So, in particular we find that by correctly calibrating the argument to $\hat{Q}_t(\cdot)$ we can achieve either approximate or exact marginal coverage.

To perform this calibration we will use a simple online update. This update proceeds by examining the empirical miscoverage frequency of the previous prediction sets and then decreasing (resp. increasing) our estimate of $\alpha^*_t$ if the prediction sets were historically under-covering (resp. over-covering) $Y_t$. In particular, let $\alpha_1$ denote our initial estimate (in our experiments we will choose $\alpha_1 = \alpha$). Recursively define the sequence of miscoverage events 
\[
\text{err}_t := \begin{cases}
1, \text{ if } Y_{t} \notin \hat{C}_{t}(\alpha_t),\\
0 \text{, otherwise,}
\end{cases}\text{ where } \hat{C}_{t}(\alpha_t) := \{y : S_{t}(X_t,y) \leq \hat{Q}_t(1-\alpha_t) \}.
\]
Then, fixing a step size parameter $\gamma > 0$ we consider the simple online update
\begin{equation}\label{eq:simple_alpha_update}
\alpha_{t+1} := \alpha_t + \gamma(\alpha - \text{err}_t).
\end{equation}
We refer to this algorithm as \textit{adaptive conformal inference}. Here, $\text{err}_t$ plays the role of our estimate of the historical miscoverage frequency. A natural alternative to this is the update
\begin{equation}\label{eq:local_weighted_alpha_update}
\alpha_{t+1} = \alpha_t + \gamma\left(\alpha - \sum_{s=1}^t w_s \text{err}_s \right),
\end{equation}
where $\{w_s\}_{1 \leq s \leq t} \subseteq [0,1]$ is a sequence of increasing weights with $\sum_{s=1}^t w_s=1$. This update has the appeal of more directly evaluating the recent empirical miscoverage frequency when deciding whether or not to lower or raise $\alpha_t$. In practice, we find that (\ref{eq:simple_alpha_update}) and (\ref{eq:local_weighted_alpha_update}) produce almost identical results. For example, in Section \ref{sec:alpha_trajs} in the Appendix we show some sample trajectories for $\alpha_t$ obtained using the update (\ref{eq:local_weighted_alpha_update}) with 
\[
w_s := \frac{0.95^{t-s}}{\sum_{s'=1}^t 0.95^{t-s'}}.
\]
We find that these trajectories are very similar to those produced by (\ref{eq:simple_alpha_update}). The main difference is that the trajectories obtained with (\ref{eq:local_weighted_alpha_update}) are smoother with less local variation in $\alpha_t$. In the remainder of this article we will focus on (\ref{eq:simple_alpha_update}) for simplicity. 

\subsection{Choosing the step size}\label{sec:choosing_gamma}

The choice of $\gamma$ gives a tradeoff between adaptability and stability. While raising the value of $\gamma$ will make the method more adaptive to observed distribution shifts, it will also induce greater volatility in the value of $\alpha_t$. In practice, large fluctuations in $\alpha_t$ may be undesirable as it allows the method to oscillate between outputting small conservative and large anti-conservative prediction sets. 

In Theorem \ref{thm:reg_bound} we give an upper bound on $(M_t(\alpha_t) - \alpha)^2$ that is optimized by choosing $\gamma$ proportional to $\sqrt{|\alpha^*_{t+1} - \alpha^*_{t}|}$. While not directly applicable in practice, this result supports the intuition that in environments with greater distributional shift the algorithm needs to be more adapatable and thus $\gamma$ should be chosen to be larger. In our experiments we will take $\gamma = 0.005$. This value was chosen because it was found to give relatively stable trajectories for $\alpha_t$ while still being sufficiently large as to allow $\alpha_t$ to adapt to observed shifts. In agreement with the general principles outlined above we found that larger values of $\gamma$ also successfully protect against distribution shifts, while taking $\gamma$ to be too small causes adaptive conformal inference to perform similar to non-adaptive methods that hold $\alpha_t = \alpha$ constant across time.

\subsection{Real data example: predicting market volatility}\label{sec:stock_prediction}

We apply ACI to the prediction of market volatility. Let
$\{P_t\}_{1 \leq t \leq T}$ denote a sequence of daily open prices for
a stock. For all $t \geq 2$, define the return
$R_t := (P_t - P_{t-1})/P_{t-1}$ and realized volatility
$V_t = R_t^2$. Our goal is to use the previously observed returns
$X_t := \{R_s\}_{1 \leq s \leq t-1}$ to form prediction sets for
$Y_t := V_t$. More sophisticated financial models might augment $X_t$
with additional market covariates (available to the analyst at time
$t-1$). As the primary purpose of this section is to illustrate
adaptive conformal inference we work with only a simple prediction method.

We start off by forming point predictions using a GARCH(1,1) model \cite{Bollerslev1986}. This method assumes that $R_t = \sigma_t \epsilon_t$ with $\epsilon_2,\dots,\epsilon_T$ taken to be i.i.d. $\mathcal{N}(0,1)$ and $\sigma_t$ satisfying the recursive update
\[
\sigma_t^2 = \omega + \tau V_{t-1} + \beta \sigma_{t-1}^2.
\]
This is a common approach used for forecasting volatility in
economics. In practice, shifting market dynamics can cause the
predictions of this model to become inaccurate over large time
periods. Thus, when forming point predictions we fit the model using
only the last 1250 trading days (i.e. approximately 5 years) of market
data. More precisely, for all times $t > 1250$ we fit the coefficients
$\hat{\omega}_{t},\ \hat{\tau}_t, \hat{\beta}_t$ as well as the
sequence of variances $\{\hat{\sigma}^t_{s}\}_{1 \leq s \leq t-1}$
using only the data $\{R_r\}_{t-1250 \leq r < t}$. Then, our point
prediction for the realized volatility at time $t$ is 
\[
(\hat{\sigma}^t_t)^2 := \hat{\omega}_{t} +\hat{\tau}_t V_{t-1} + \hat{\beta}_t( \hat{\sigma}^t_{t-1})^2.
\]
To form prediction intervals we define the sequence of conformity scores 
\[
S_t := \frac{|V_t - (\hat{\sigma}^t_t)^2|}{(\hat{\sigma}^t_t)^2}
\]
and the corresponding quantile function
\[
\hat{Q}_t(p) := \inf\left\{ x : \frac{1}{1250} \sum_{r=t-1250}^{t-1} \bone_{S_r \leq x} \geq p\right\}. 
\]
Then, our prediction set at time $t$ is
\[
\hat{C}_{t}(\alpha_t) := \left\{v : \frac{|v - (\hat{\sigma}^t_t)^2|}{(\hat{\sigma}^t_t)^2} \leq \hat{Q}_t(1-\alpha_t) \right\},
\]
where $\{\alpha_t\}$ is initialized with $\alpha_{1250}=\alpha = 0.1$ and then updated recursively as in (\ref{eq:simple_alpha_update}). 

We compare this algorithm to a non-adaptive alternative that takes $\alpha_t = \alpha$ fixed. To measure the performance of these methods across time we examine their local coverage frequencies defined as the average coverage rate over the most recent two years, i.e.
\begin{equation}\label{eq:local_cov_freq}
\text{localCov}_t := 1-\frac{1}{500} \sum_{r=t-250+1}^{t+250} \text{err}_r.
\end{equation}
If the methods perform well then we expect the local coverage frequency to stay near the target value $1-\alpha$ across all time points.

\begin{figure}
  \centering
  \includegraphics[scale=0.35]{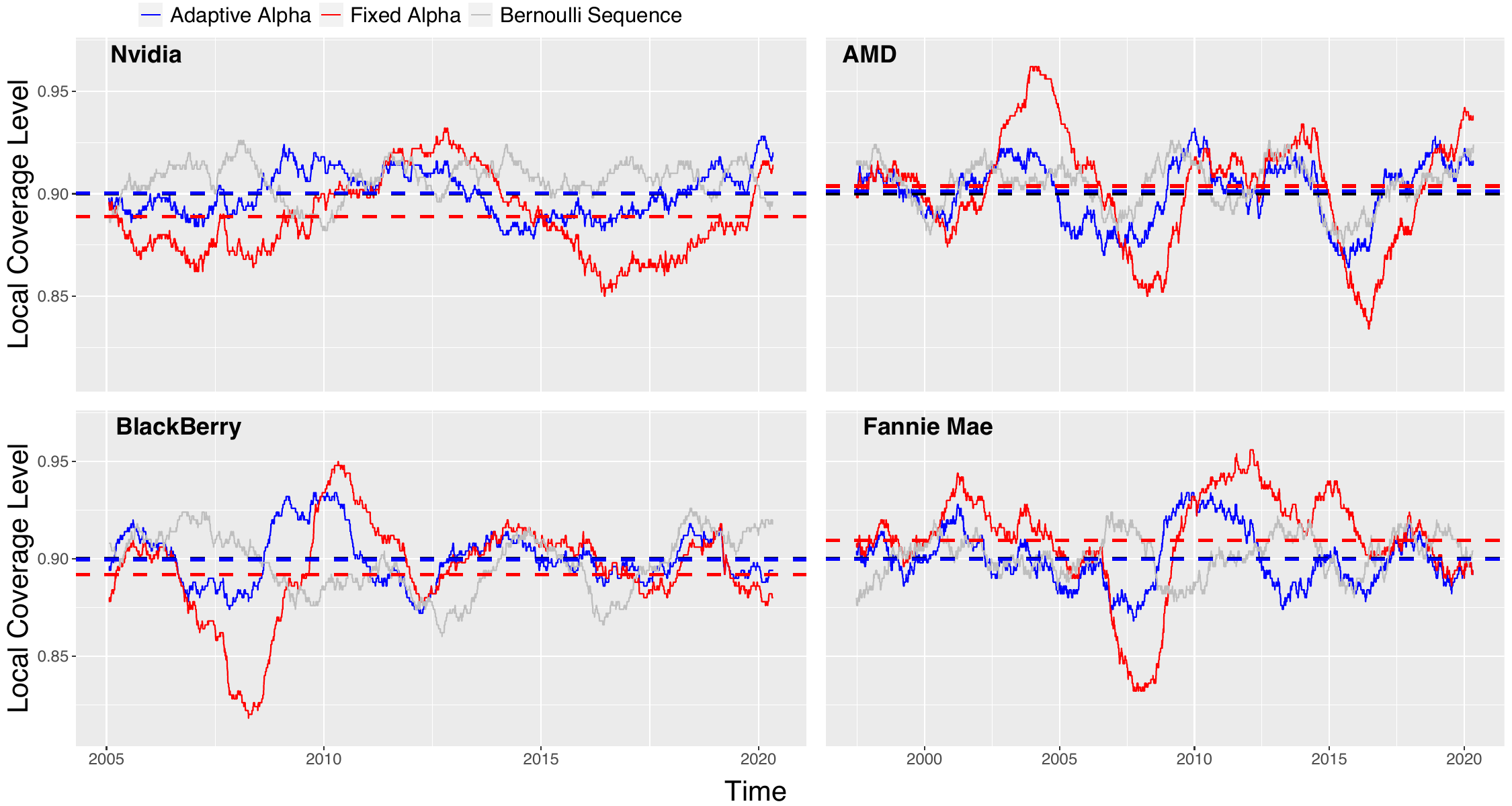}
  \caption{Local coverage frequencies for adaptive conformal (blue), a non-adaptive method that holds $\alpha_t = \alpha$ fixed (red), and an i.i.d. Bernoulli(0.1) sequence (grey) for the prediction of stock market volatility. The coloured dotted lines mark the average coverage obtained across all time points, while the black line indicates the target level of $1-\alpha=0.9$. }
  \label{fig:stock_cov}
\end{figure}

Daily open prices were obtained from publicly available datasets published by \textit{The Wall Street Journal}. The realized local coverage frequencies for the non-adaptive and adaptive conformal methods on four different stocks are shown in Figure \ref{fig:stock_cov}. These stocks were selected out of a total of 12 stocks that we examined because they showed a clear failure of the non-adaptive method. Adaptive conformal inference was found to perform well in all cases (see Figure \ref{fig:additional_stocks} in the appendix).

As a visual comparator, the grey curves show the moving average $1-\frac{1}{500} \sum_{r=t-250+1}^{t+250} I_r$ for sequences $\{I_t\}_{1 \leq t \leq T}$ that are i.i.d.~Bernoulli(0.1). We see that the local coverage frequencies obtained by adaptive conformal inference (blue lines) always stay within the variation that would be expected from an i.i.d. Bernoulli sequence. On the other hand, the non-adaptive method undergoes large excursions away from the target level of $1-\alpha = 0.9$ (red lines). For example, in the bottom right panel we can see that the non-adaptive method fails to cover the realized volatility of Fannie Mae during the 2008 financial crisis, while the adaptive method is robust to this event (see Figure \ref{fig:stock_prices} in the Appendix for a plot of the price of Fannie Mae over this time period).

\section{Related Work}

Prior work on conformal inference has considered two different types of distribution shift \cite{Tibs2019, Maxime2020}. In both cases the focus was on environments in which the calibration data is drawn i.i.d. from a single distribution $P_0$, while the test point comes from a second distribution $P_1$. In this setting \citet{Tibs2019} showed that valid prediction sets can be obtained by re-weighting the calibration data using the likelihood ratio between $P_1$ and $P_0$. However, this requires the conditional distribution of $Y|X$ to be constant between training and testing and the likelihood ratio $P_{1}(X)/P_{0}(X)$ to be either known or very accurately estimated. On the other hand, \citet{Maxime2020} develop methods for forming prediction sets that are valid whenever $P_1$ and $P_0$ are close in $f$-divergence. Similar to our work, they show that if $D_f(P_1||P_0) \leq \rho$ then there exists a conservative value $\alpha_{\rho} \in (0,1)$ such that 
\[
M(\alpha_{\rho}) := \mmp(S(X_{t},Y_{t}) > \hat{Q}(1-\alpha_{\rho}))  \leq \alpha.
\]
The difference between our approach and theirs is twofold. First, while they fix a single conservative value $\alpha_{\rho}$ our methods aim to estimate the optimal choice $\alpha^*$ satisfying $M(\alpha^*) = \alpha$. This is not possible in the setting of \cite{Maxime2020} as they do not observe any data from which the size of the distribution shift can be estimated. Second, while they consider only one training and one testing distribution we work in a fully online setting in which the distribution is allowed to shift continuously over time.   

\section{Coverage guarantees}\label{sec:theory}

\subsection{Distribution-free results}\label{sec:agnostic_bounds}

In this section we outline the theoretical coverage guarantees of adaptive conformal inference. We will assume throughout that with probability one $\alpha_1 \in [0,1]$ and $\hat{Q}_t$ is non-decreasing with $\hat{Q}_t(x)  = -\infty$ for all $x < 0$ and $\hat{Q}_t(x) = \infty$ for all $x>1$. Our first result shows that over long time intervals adaptive conformal inference obtains the correct coverage frequency irrespective of any assumptions on the data-generating distribution.

\begin{lemma}\label{lem:alphat_is_bounded}
With probability one we have that $\forall t \in \mmn$, $\alpha_t \in [-
\gamma,1+\gamma]$.
\end{lemma}

\begin{proof}
Assume by contradiction that with positive probability $\{\alpha_t\}_{t \in \mmn}$ is such that $\inf_t \alpha_t < -\gamma$ (the case where $\sup_t \alpha_t > 1+\gamma$ is identical). Note that $\sup_{t}|\alpha_{t+1} - \alpha_{t}| = \sup_{t}\gamma|\alpha - \text{err}_t| < \gamma$. Thus, with positive probability we may find $t \in \mmn$ such that $\alpha_t <0$ and $\alpha_{t+1} < \alpha_t$. However,
\[
\alpha_{t} < 0 \implies \hat{Q}_t(1-\alpha_t) = \infty \implies \text{err}_{t} = 0 \implies \alpha_{t+1} = \alpha_{t} + \gamma(\alpha - \text{err}_{t}) \geq \alpha_{t}
\] 
and thus $\mmp(\exists t \text{ such that } \alpha_{t+1} < \alpha_t < 0) =0$. We have reached a contradiction.
\end{proof}

\begin{proposition}\label{prop:bound_on_errs} 
 With probability one we have that for all $T \in \mmn$,
\begin{equation}\label{eq:bound_on_errs}
\left| \frac{1}{T} \sum_{t=1}^T \text{err}_t - \alpha \right| \leq \frac{\max\{\alpha_1,1-\alpha_1\}+\gamma}{T\gamma}.
\end{equation}
In particular, $\lim_{T \to \infty} \frac{1}{T}\sum_{t=1}^T \text{err}_t \stackrel{a.s.}{=}\alpha.$
\end{proposition}
\begin{proof}
By expanding the recursion defined in (\ref{eq:simple_alpha_update}) and applying Lemma \ref{lem:alphat_is_bounded} we find that
\[
[-\gamma,1+\gamma]  \ni \alpha_{T+1} = \alpha_1 + \sum_{t=1}^{T} \gamma(\alpha - \text{err}_t).
\]
Rearranging this gives the result.
\end{proof}

Proposition \ref{prop:bound_on_errs} puts no constraints on the
data generating distribution. One may immediately ask whether these
results can be improved by making mild assumptions on the distribution
shifts. We argue that without assumptions on the quality of the
initialization the answer to this question is negative. To understand this, consider a setting in which there is a single fixed optimal target $\alpha^* \in [0,1]$ and assume that
\[
M_t(p) = M(p)  =  \begin{cases}
\alpha + \frac{1-\alpha}{1-\alpha^*}(p-\alpha^*), \text{ if } p > \alpha^*,\\
\alpha + \frac{\alpha}{\alpha^*}(p-\alpha^*) \text{ if } p \leq \alpha^*.
\end{cases}.
\]
Suppose additionally that $\mme[\text{err}_t|\alpha_t] = M(\alpha_t)$.\footnote{This last assumption is in general only true if $\hat{Q}_t(\cdot)$ and $\alpha_t$ are fit independently of one another.}   In order to simplify the calculations consider the noiseless update $\alpha_{t+1} = \alpha_t + \gamma(\alpha-M(\alpha_t)) =  \alpha_t + \gamma(\alpha-\mme[\text{err}_t|\alpha_t])$. Intuitively, the noiseless update can be viewed as the average case behaviour of (\ref{eq:simple_alpha_update}). Now, for any initialization $\alpha_1$ and any $\gamma \leq \min\{ \frac{1-\alpha^*}{1-\alpha} ,  \frac{\alpha^*}{\alpha}\}$ there exists a constant $c \in \{ \frac{1-\alpha}{1-\alpha^*} ,  \frac{\alpha}{\alpha^*}\}$ such that for all $t$, $M(\alpha_t) - \alpha = c(\alpha_t -\alpha^*)$. So, we have that 
\[
\mme[\text{err}_t] - \alpha = c\mme[\alpha_t - \alpha^*] = c\mme[\alpha_{t-1} + \gamma(\alpha-M_{t-1}(\alpha_{t-1})) - \alpha^*] = c(1-c\gamma) \mme[\alpha_{t-1} -\alpha^*].
\]
Repeating this calculation recursively gives that 
\[
\mme[\text{err}_t]   - \alpha  = c(1-c\gamma)^{t-1}\mme[\alpha_1- \alpha^*] = c(1-c\gamma)^{t-1}(\alpha_1 - \alpha^*),
\] 
and thus,
\[
\left| \frac{1}{T} \sum_{t=1}^T \mme[\text{err}_t] - \alpha \right| = \frac{1- (1-c\gamma)^{T}}{T\gamma} |\alpha_1 - \alpha^*|.
\]
The comparison of this bound to (\ref{eq:bound_on_errs}) is self-evident. The main difference is that we have replaced $\max\{1-\alpha_1,\alpha_1\}$ with $|\alpha_1 - \alpha^*|$. This arises from the fact that $\alpha^* \in (0,1)$ is arbitrary and thus $\max\{1-\alpha_1,\alpha_1\}$  is the best possible upper bound on $|\alpha_1 - \alpha^*|$. So, we view Proposition \ref{prop:bound_on_errs} as both an agnostic guarantee that shows that our method gives the correct long-term empirical coverage frequency irrespective of the true data generating process, and as an approximately tight bound on the worst-case behaviour immediately after initialization. 

\subsection{Performance in a hidden Markov model}\label{sec:MC_converage_bounds}

Although we believe Proposition \ref{prop:bound_on_errs} is an approximately tight characterization of the behaviour after initialization, we can still ask whether better bounds can be obtained for large time steps. In this section we answer this question positively by showing that if $\alpha_1$ is initialized appropriately and the distribution shift is small, then tighter coverage guarantees can be given. In order to obtain useful results we will make some simplifying assumptions about the data generating process. While we do not expect these assumptions to hold exactly in any real-world setting, we do consider our results to be representative of the true behaviour of adaptive conformal inference and we expect similar results to hold under alternative models.

\subsubsection{Setting}

We model the data as coming from a hidden Markov model. In particular, we let $\{A_t\}_{t \in \mmn} \subseteq \mathcal{A}$ denote the underlying Markov chain for the environment and we assume that conditional on $\{A_t\}_{t \in \mmn}$, $\{(X_t,Y_t)\}_{t \in \mmn}$ is an independent sequence with $(X_t,Y_t) \sim P_{A_t}$ for some collection of distributions $\{P_a:a \in \mathcal{A}\}$. In order to simplify our calculations, we assume additionally that the estimated quantile function $\hat{Q}_{t}(\cdot)$ and score function $S_t(\cdot)$ do not depend on $t$ and we denote them by $\hat{Q}(\cdot)$ and $S(\cdot)$. This occurs for example in the split conformal setting with fixed training and calibration sets. 

In this setting, $\{(\alpha_t,A_t)\}_{t \in \mmn}$ forms a Markov
chain on $[-\gamma,1+\gamma] \times \mathcal{A}$. We assume that this
chain has a unique stationary distribution $\pi$ and that
$(\alpha_1,A_1) \sim \pi$. This implies that
$(\alpha_t,A_t,\text{err}_t)$ is a stationary process and thus will
greatly simplify our characterization of the behaviour of
$\text{err}_t$. While there is little doubt that the theory can be
extended, recall our that main goal is to get useful and simple results.
That said, what we really have in mind here is that
$\{A_t\}_{ t \in \mmn}$ is sufficiently well-behaved to guarantee that
$(\alpha_t,A_t)$ has a limiting stationary distribution. In Section
\ref{sec:MC_stationary_dist} we give an example where this is indeed
provably the case. Lastly, the assumption that
$(\alpha_1,A_1) \sim \pi$ is essentially equivalent to assuming that
we have been running the algorithm for long enough to exit the
initialization phase described in Section \ref{sec:agnostic_bounds}.

\subsubsection{Large deviation bound for the errors}

Our first observation is that $\text{err}_t$ has the correct average value. More precisely, by Proposition \ref{prop:bound_on_errs} we have that $\lim_{T \to \infty} T^{-1} \sum_{t=1}^T \text{err}_t \stackrel{a.s.}{=} \alpha$ and since $\text{err}_t$ is stationary it follows that $\mme[\text{err}_t] = \alpha$. Thus, to understand the deviation of $T^{-1}\sum_{t=1}^T \text{err}_t $ from $\alpha$ we simply need to characterize the dependence structure of $\{\text{err}_t\}_{t \in \mmn}$. 

We accomplish this in Theorem \ref{thm:large_dev_err_bound}, which
gives a large deviation bound on
$|T^{-1}\sum_{t=1}^T \text{err}_t - \alpha|$. The idea behind this
result is to decompose the dependence in
$\{\text{err}_t\}_{t \in \mmn}$ into two parts. First, there is
dependence due to the fact that $\alpha_t$ is a function of
$\{\text{err}_r\}_{1 \leq r \leq t-1}$. In Section
\ref{sec:large_deviation_proof} in the Appendix we argue that this
dependence induces a negative correlation and thus the errors
concentrate around their expectation at a rate no slower than that of
an i.i.d. Bernoulli sequence. This gives rise to the first term in
(\ref{eq:large_dev_bound}), which is what would be obtained by
applying Hoeffding's inequality to an i.i.d. sequence. Second, there
is dependence due to the fact that $A_t$ depends on $A_{t-1}$. More
specifically, consider a setting in which the distribution of $Y|X$
has more variability in some states than others. The goal of adaptive
conformal inference is to adapt to the level of variability and thus
return larger prediction sets in states where the distribution of
$Y|X$ is more spread. However, this algorithm is not perfect and as a
result there may be some states $a \in \mathcal{A}$ in which
$\mme[\text{err}_t|A_t=a]$ is biased away from $\alpha$. Furthermore,
if the environment tends to spend long stretches of time in more
variable (or less variable) states this will induce a positive
dependence in the errors and cause $T^{-1}\sum_{t=1}^T \text{err}_t$
to deviate from $\alpha$. To control this dependence we use a
Bernstein inequality for Markov chains to bound
$|T^{-1}\sum_{t=1}^T \mme[\text{err}_t|A_t] - \alpha|$. This gives
rise to the second term in (\ref{eq:large_dev_bound}).

\begin{theorem}\label{thm:large_dev_err_bound}
Assume that $\{A_t\}_{t \in \mmn}$ has non-zero absolute spectral gap $1-\eta > 0$. Let 
\[
 B :=  \sup_{a \in \mathcal{A}} |\mme[\textup{err}_t|A_t=a] - \alpha|  \ \ \ \text{ and } \ \ \  \sigma^2_B := \mme[(\mme[\textup{err}_t|A_t] - \alpha)^2].
\]
Then,
\begin{equation}\label{eq:large_dev_bound}
\mmp\left(\left|\frac{1}{T}\sum_{t=1}^T \textup{err}_t - \alpha\right| \geq \epsilon\right) \leq 2\exp\left( -  \frac{T\epsilon^2}{8} \right) + 2 \exp\left( - \frac{T(1-\eta)\epsilon^2}{8(1+\eta)\sigma^2_B + 20 B\epsilon} \right) .
\end{equation}
\end{theorem}

A formal proof of this result can be found in Section \ref{sec:large_deviation_proof}. The quality of this concentration inequality will depend critically on the size of the bias terms $B$ and $\sigma^2_B$. Before proceeding, it is important that we emphasize that the definitions of $B$ and $\sigma^2_B$ are independent of the choice of $t$ owing to the fact that $(\alpha_t,A_t,\text{err}_t)$ is assumed stationary. Now, to understand these quantities, let
\[
M(p|a) := \mmp(S(X_t,Y_t) > \hat{Q}(1-p) | A_t = a)
\] 
denote the realized miscoverage level in state $a \in \mathcal{A}$ obtained by the quantile $\hat{Q}(1-p)$. Assume that $M(p|a)$ is continuous. This will happen for example when $\hat{Q}(\cdot)$ is continuous and $S(X_t,Y_t)|A_t=a$ is continuously distributed. Then, there exists an optimal value $\alpha^*_a$ such that $M(\alpha^*_a|a) = \alpha$. Lemma \ref{lem:bound_on_mean_given_a} in the Appendix shows that if in addition $M(\cdot|a)$ admits a second order Taylor expansion, then 
\[
B \leq C\left(\gamma + \gamma^{-1} \sup_{a \in \mathcal{A}} \sup_{k \in \mmn} \mme[|\alpha^*_{A_{t+1}} - \alpha^*_{A_t}| \big|A_{t+k} = a] \right) \ \ \ \text{ and } \ \ \ \sigma^2_B \leq B^2.
\]
Here, the constant $C$ will depend on how much $M(\cdot|a)$ differs from the ideal case in which $\hat{Q}(\cdot)$ is the true quantile function for $S(X_t,Y_t)|A_t=a$. In this case we would have that $M(\cdot|a)$ is the linear function $M(p|a) =  p$, $\forall p \in [0,1]$ and $C \leq 2$.

We remark that the term $\mme[|\alpha^*_{A_{t+1}} - \alpha^*_{A_t}| \big|A_{t+k} = a]$ can be seen as a quantitative measurement of the size of the distribution shift in terms of the change in the critical value $\alpha^*_a$. Thus, we interpret these results as showing that if the distribution shift is small and $\forall a \in \mathcal{A}$, $\hat{Q}(\cdot)$ gives reasonable coverage of the distribution of $S(X_t,Y_t)|A_t=a$, then $T^{-1}\sum_{t=1}^T \text{err}_t$ will concentrate well around $\alpha$. 

\subsubsection{Achieving approximate marginal coverage}

Theorem \ref{thm:large_dev_err_bound} bounds the distance between the average miscoverage rate and the target level over long stretches of time. On the other hand, it provides no information about the marginal coverage frequency at a single time step. The following result shows that if the distribution shift is small, the realized marginal coverage rate $M(\alpha_t|A_t)$ will be close to $\alpha$ on average.
\begin{theorem}\label{thm:reg_bound}
Assume that there exists a constant $L>0$ such that for all $a \in \mathcal{A}$ and all $\alpha_1,\alpha_2 \in \mmr$,
\[
|M(\alpha_2|a) - M(\alpha_1|a)| \leq L|\alpha_2 - \alpha_1|.
\]
Assume additionally that for all $a \in \mathcal{A}$ there exists $\alpha^*_a \in (0,1)$ such that $M(\alpha^*_a|a) = \alpha$. Then,
\begin{equation}\label{eq:regret_bound}
 \mme[(M(\alpha_{t}|A_t) - \alpha)^2] \leq \frac{L(1+\gamma)}{\gamma} \mme[|\alpha^*_{A_{t+1}} - \alpha^*_{A_t}|] + \frac{L}{2}\gamma.
\end{equation}
\end{theorem}

Once again we emphasize that (\ref{eq:regret_bound}) holds for any choice of $t$ owing to the fact that $(\alpha_t,A_t,\text{err}_t)$ is assumed stationary and thus the quantities appearing in the bound are invariant across $t$. Proof of this result can be found in Section \ref{sec:approx_marginal_proof} of the Appendix. We remark that the right-hand side of (\ref{eq:regret_bound}) is minimized by choosing $\gamma = (2\mme[|\alpha^*_{A_{t+1}} - \alpha^*_{A_t}|])^{1/2}$, which gives the inequality
\[
 \mme[(M(\alpha_{t}|A_t) - \alpha)^2] \leq L(\sqrt{2}+1)\sqrt{\mme[|\alpha^*_{A_{t+1}} - \alpha^*_{A_t}|]}.
\]
As above we have that in the ideal case $\hat{Q}(\cdot)$ is a perfect estimate of the quantiles of $S(X_t,Y_t)|A_t=a$ and thus $M(p|a) = p$ and $L=1$. Moreover, we once again have the interpretation that $\mme[|\alpha^*_{A_{t+1}} - \alpha^*_{A_t}|]$ is a quantitative measurement of the distribution shift. Thus, this result can be interpreted as bounding the average difference between the realized and target marginal coverage in terms of the size of the underlying distribution shift. Finally, note that the choice $\gamma = (2\mme[|\alpha^*_{A_{t+1}} - \alpha^*_{A_t}|])^{1/2}$ formalizes our intuition that $\gamma$ should be chosen to be larger in domains with greater distribution shift, while not being so large as to cause $\alpha_t$ to be overly volatile.

\section{Impact of $\pmb{S_t}(\cdot)$ on the performance}

The performance of all conformal inference methods depends heavily on the design of the conformity score. Previous work has shown how carefully chosen  scores or even explicit optimization of the interval width can be used to obtain smaller prediction sets \cite[e.g.][]{Romano2019,Romano2020,Kivaranovic2020Adaptive,Chen2021}. Adaptive conformal inference can work with any conformity score $S_t(\cdot)$ and quantile function $\hat{Q}_t(\cdot)$ and thus can be directly combined with other improvements in conformal inference to obtain shorter intervals. One important caveat here is that the lengths of conformal prediction sets depend directly on the quality of the fitted regression model. Thus, to obtain smaller intervals one should re-fit the model at each time step using the most recent data to build the most accurate predictions. This is exactly what we have done in our experiments in Sections \ref{sec:stock_prediction} and \ref{sec:election_forecasting}.

In addition to this, the choice of $S_t(\cdot)$ can also have a direct effect on the coverage properties of adaptive conformal inference. Theorems  \ref{thm:large_dev_err_bound} and \ref{thm:reg_bound} show that the performance of adaptive conformal inference is controlled by the size of the shift in the optimal parameter $\alpha^*_t$ across time. Moreover, $\alpha^*_t$ itself is in one-to-one correspondence with the $1-\alpha$ quantile of $S_t(X_t,Y_t)$. Thus, the coverage properties of adaptive conformal inference depend on how close $S_t(X_t,Y_t)$ is to being stationary.

For a simple example illustrating the impact of this dependence, note that in Section \ref{sec:stock_prediction} we formed prediction sets using the conformity score
\[
S_t := \frac{|V_t - \hat{\sigma}^2_t|}{\hat{\sigma}^2_t}.
\]
An \textit{a priori} reasonable alternative to this is the unnormalized score 
\[
\tilde{S}_t := |V_t - \hat{\sigma}^2_t|.
\]
However, after a more careful examination it becomes unsurprising that normalization by $\hat{\sigma}^2_t$ is critical for obtaining an approximately stationary conformity score and thus $\tilde{S}_t$ leads to much worse coverage properties. Figure \ref{fig:stock_cov_bad_st} shows the local coverage frequency (see (\ref{eq:local_cov_freq})) of adaptive conformal inference using $\tilde{S}_t$. In comparison to Figure \ref{fig:stock_cov} the coverage now undergoes much wider swings away from the target level of 0.9. This issue can be partially mitigated by choosing a larger value of $\gamma$ that gives greater adaptivity to the algorithm.

\begin{figure}
\centering
\includegraphics[scale=0.36]{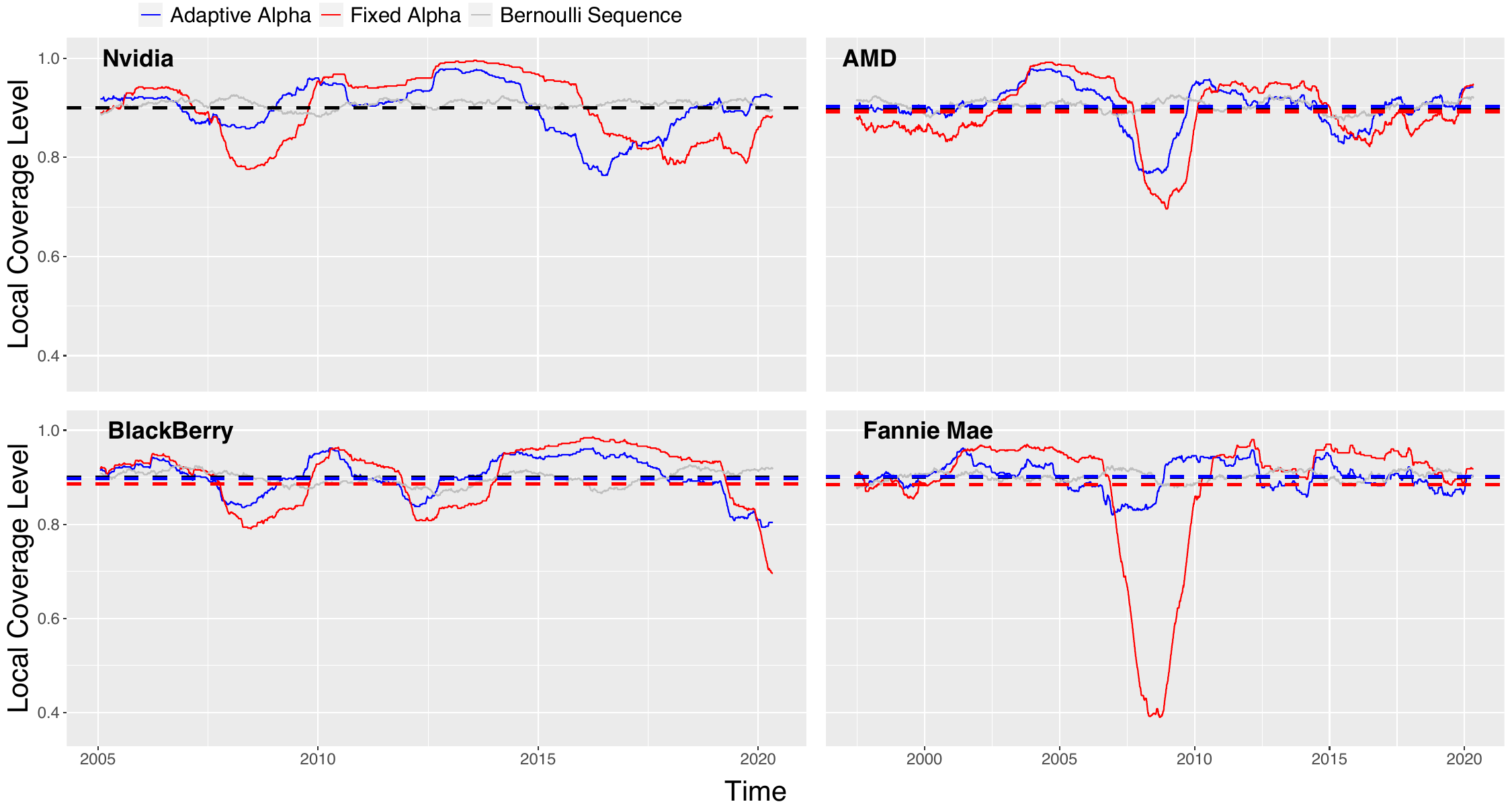}
 \caption{Local coverage frequencies for adaptive conformal (blue), a non-adaptive method that holds $\alpha_t = \alpha$ fixed (red), and an i.i.d. Bernoulli(0.1) sequence (grey) for the prediction of stock market volatility with conformity score $\tilde{S}_t$. The coloured dotted lines mark the average coverage obtained across all time points, while the black line indicates the target level of $1-\alpha=0.9$. }
    \label{fig:stock_cov_bad_st}
\end{figure}

\section{Real data example: election night predictions}\label{sec:election_forecasting}

During the 2020 US presidential election \textit{The Washington Post} used conformalized quantile regression (CQR) (see (\ref{eq:split_quantile}) and Section \ref{sec:conformal_inference}) to produce county level predictions of the vote total on election night \cite{Cherian2020}. Here we replicate the core elements of this method using both fixed and adaptive quantiles. 

To make the setting precise, let $\{Y_t\}_{1 \leq t \leq T}$ denote
the number of votes cast for presidential candidate Joe Biden in the 2020
election in each of approximately $T=3000$ counties in the United
States. Let $X_t$ denote a set of demographic covariates associated to
the $t$th county. In our experiment $X_t$ will include information on
the make-up of the county population by ethnicity, age, sex, median
income and education (see Section \ref{sec:election_dat} for
details). On election night county vote totals were observed as soon
as the vote count was completed. If the order in which vote counts
completed was uniformly random $\{(X_t,Y_t)\}_{1 \leq t \leq T}$
would be an exchangeable sequence on which we could run standard
conformal inference methods. In reality, larger urban counties tend to
report results later than smaller rural counties and counties on the
east coast of the US report earlier than those on the west coast. Thus, the
distribution of $(X_t,Y_t)$ can be viewed as drifting throughout
election night.

\begin{figure}
  \centering
  \includegraphics[scale=0.34]{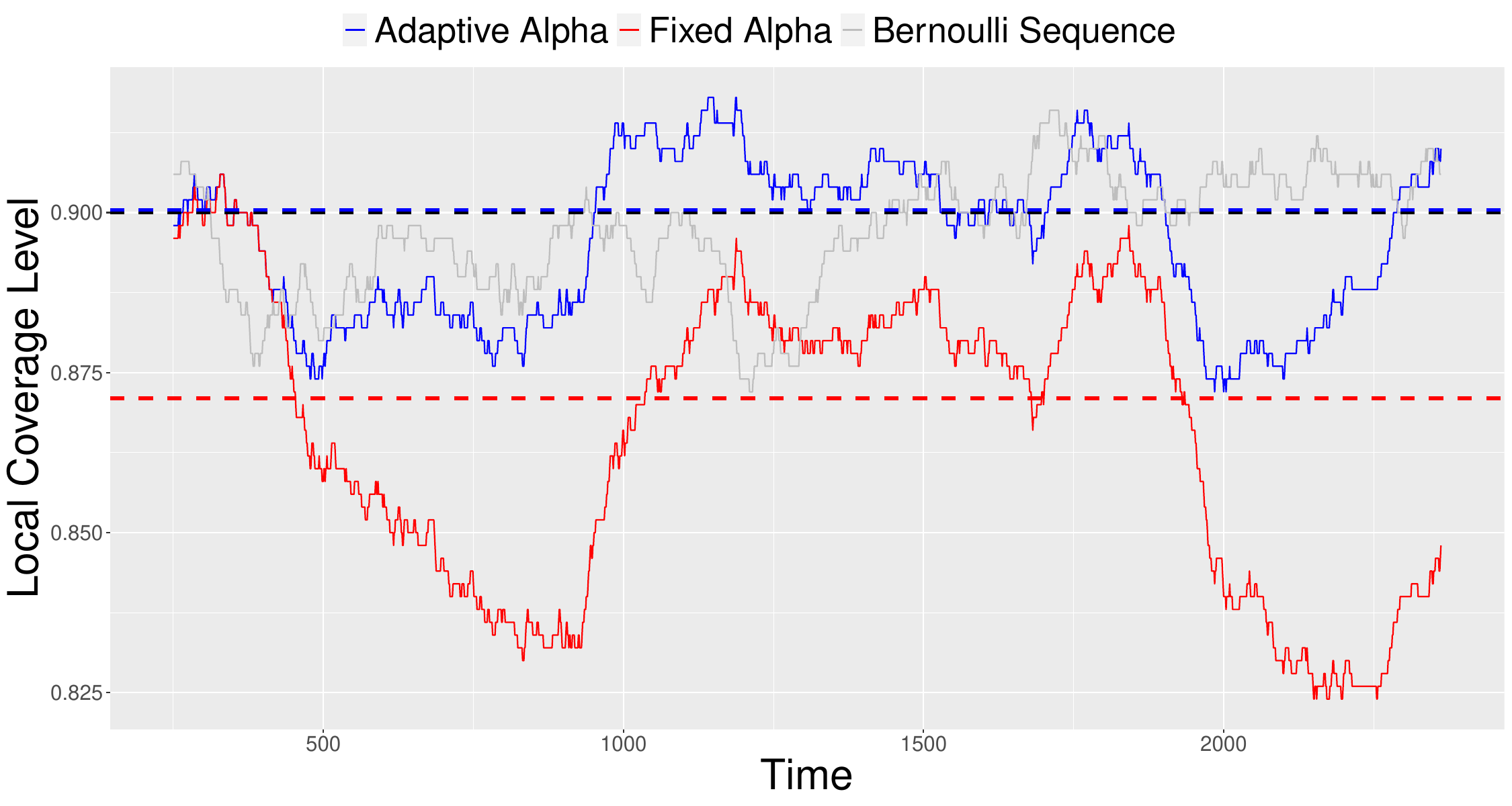}
  \caption{Local coverage frequencies of adaptive conformal (blue), a non-adaptive method that holds $\alpha_t = \alpha$ fixed (red), and an i.i.d. Bernoulli(0.1) sequence (grey) for county-level election predictions. Coloured dotted lines show the average coverage across all time points, while the black line indicates the target coverage level of $1-\alpha=0.9$. }
  \label{fig:election_coverage}
\end{figure}

 We apply CQR to predict the county-level vote totals (see Section \ref{sec:cqr_alg} for details). To replicate the east to west coast bias observed on election night we order the counties by their time zone with eastern time counties appearing first and Hawaiian counties appearing last. Within each time zone counties are ordered uniformly at random. Figure \ref{fig:election_coverage} shows the realized local coverage frequency over the most recent 300 counties (see (\ref{eq:local_cov_freq})) for the non-adaptive and adaptive conformal methods. We find that the non-adaptive method fails to maintain the desired $90\%$ coverage level, incurring large troughs in its coverage frequency during time zone changes. On the other hand, the adaptive method maintains approximate $90\%$ coverage across all time points with deviations in its local coverage level comparable to what is observed in Bernoulli sequences.  

\section{Discussion}\label{sec:limitations}

There are still many open problems in this area. The methods we
develop are specific to cases where $Y_t$ is revealed at each time
point. However, there are many settings in which we receive the
response in a delayed fashion or in large batches. In addition, our
theoretical results in Section \ref{sec:MC_converage_bounds} are
limited to a single model for the data generating distribution and the
special case where the quantile function $\hat{Q}_t(\cdot)$ is fixed across
time. It would be interesting to determine if similar results can be
obtained in settings where $\hat{Q}_t(\cdot)$ is fit in an online
fashion on the most recent data. Another potential area for
improvement is in the choice of the step size $\gamma$. In Section
\ref{sec:choosing_gamma} we give some heuristic guidelines for
choosing $\gamma$ based on the size of the distribution shift in the
environment. Ideally however we would like to be able to determine
$\gamma$ adaptively without prior knowledge. Finally, our experimental
results are limited to just two domains. Additional work is needed to
determine if our methods can successfully protect against a wider
variety of real-world distribution shifts.

\section{Acknowledgements}

E.C. was supported by Office of Naval Research grant N00014-20-12157, by the National Science Foundation grants OAC 1934578 and DMS 2032014, by the Army Research Office (ARO) under grant W911NF-17-1-0304, and by the Simons Foundation under award 814641. We thank John Cherian for valuable discussions related to Presidential Election Night 2020 and Lihua Lei for helpful comments on the relationship to online learning.

\bibliography{OnlineConformalPredicitonBib.bib}
\bibliographystyle{plainnat}

\newpage


\appendix

\section{Appendix}

\subsection{Connection to online learning}

In Section \ref{sec:main_method} we motivated the update (\ref{eq:simple_alpha_update}) as a way to adjust the size of our prediction sets in response to the realized historical miscoverage frequency. Alternatively, one could also derive (\ref{eq:simple_alpha_update}) as an online gradient descent algorithm with respect to the pinball loss. To be more precise let
\[
\beta_t := \sup\{\beta : Y_t \in \hat{C}_t(\beta)\},
\]
where we remark that $\hat{C}_t(\beta_t)$ can be thought of as the smallest prediction set containing $Y_t$. Additionally, define the pinball loss
\[
 \rho_{\alpha}(u) = \begin{cases}
 \alpha u,\ u > 0,\\
 -(1- \alpha)u,\ u\leq 0.\\
 \end{cases}.
\]
and consider the loss $\ell(\alpha_t,\beta_t) = \rho_{\alpha}(\beta_t - \alpha_t)$. Then, one directly computes that 
\[
\alpha_t -  \gamma \partial_{\alpha_t} \ell(\alpha_t,\beta_t) = \alpha_t +  \gamma(\alpha - \bone_{\alpha_t > \beta_t}) =  \alpha_t +  \gamma(\alpha - \text{err}_t).
\]
Because the pinball loss is convex, this gradient descent update falls within a well understood class of algorithms that have been extensively studied in the online learning literature (see e.g. \cite{Hazan2019}). A standard analysis may then be to bound the regret of $\alpha_t$ defined as
\[
\text{Reg}_T := \sum_{t=1}^T \ell(\alpha_t,\beta_t) - \min_{\beta} \sum_{t=1}^T \ell(\beta,\beta_t).
\]    
Unfortunately, this notion of regret fails to capture our intuition that $\alpha_t$ is adaptively tracking the moving target $\alpha^*_t$. Thus, we make the connection to online gradient descent only in passing and develop alternative theoretical tools in Section \ref{sec:theory}.

\subsection{Stock prices}

Figure \ref{fig:stock_prices} shows daily open prices for the four stocks considered in Section \ref{sec:stock_prediction}.

\begin{figure}[H] 
  \centering
  \includegraphics[scale=0.38]{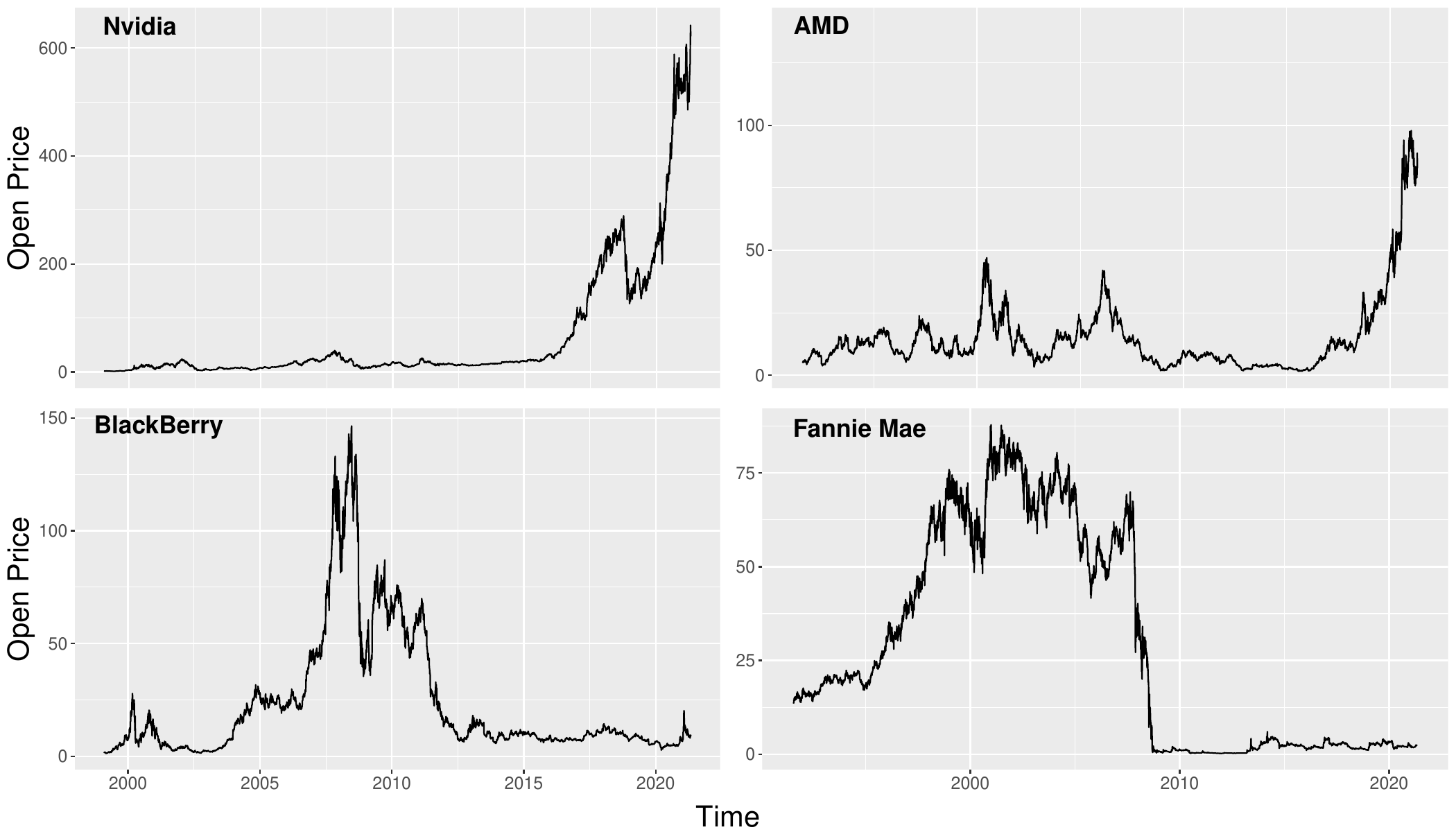}
  \caption{Daily open prices for the four stocks considered in Section \ref{sec:stock_prediction}.}
  \label{fig:stock_prices}
\end{figure}

\subsection{Trajectories of $\pmb{\alpha_t}$}\label{sec:alpha_trajs}

In this section we show the realized trajectories of $\alpha_t$ obtained in our experiments from Sections \ref{sec:stock_prediction} and \ref{sec:election_forecasting}. 

 \begin{figure}[H]
  \centering
  \includegraphics[scale=0.38]{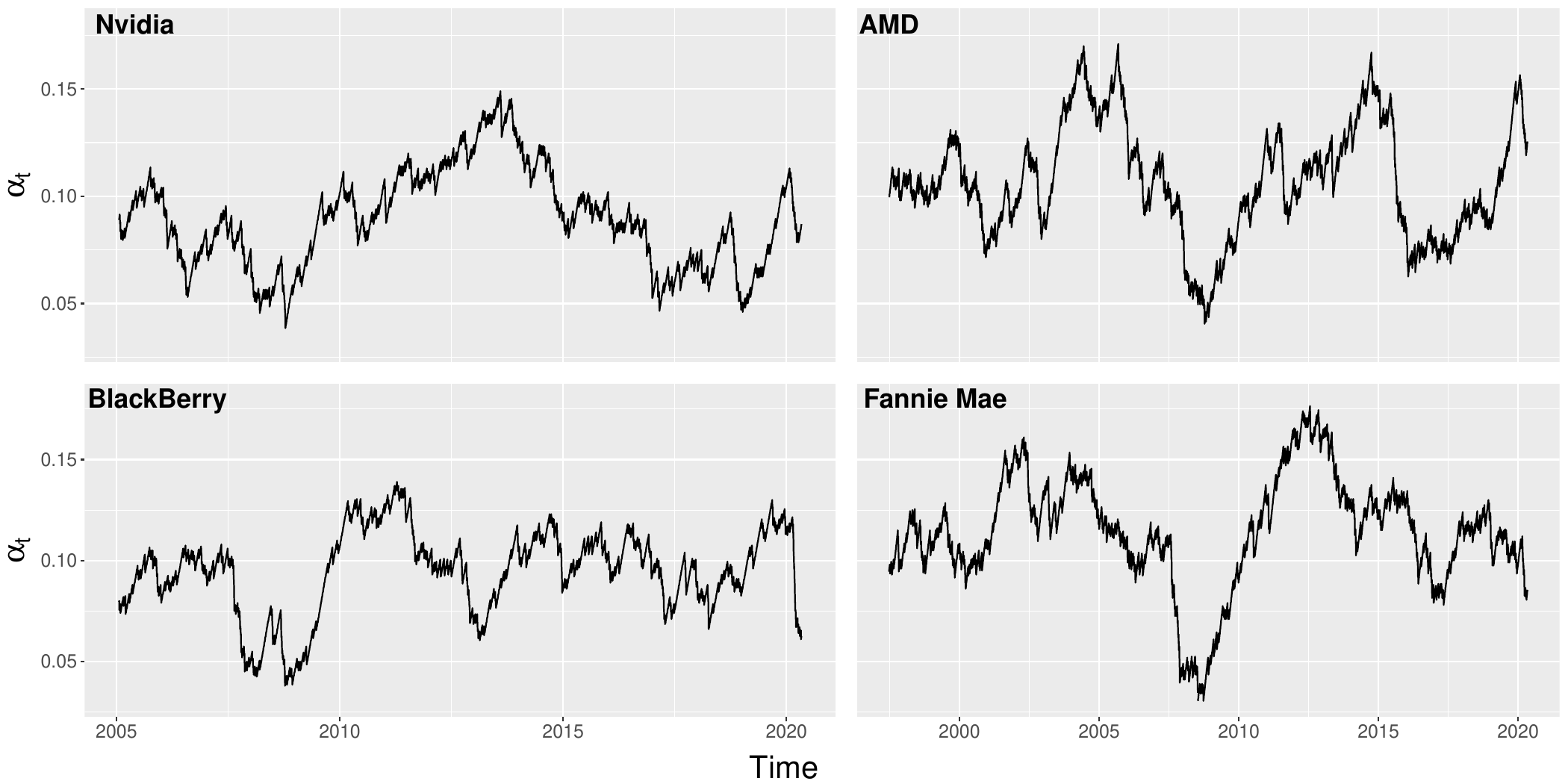}
  \caption{Realized trajectories of $\alpha_t$ for predicting stock market volatility as outlined in Section \ref{sec:stock_prediction} using update (\ref{eq:simple_alpha_update}).}
\label{fig:stock_alpha_traj_simple}
\end{figure}

 \begin{figure}[H]
  \centering
  \includegraphics[scale=0.38]{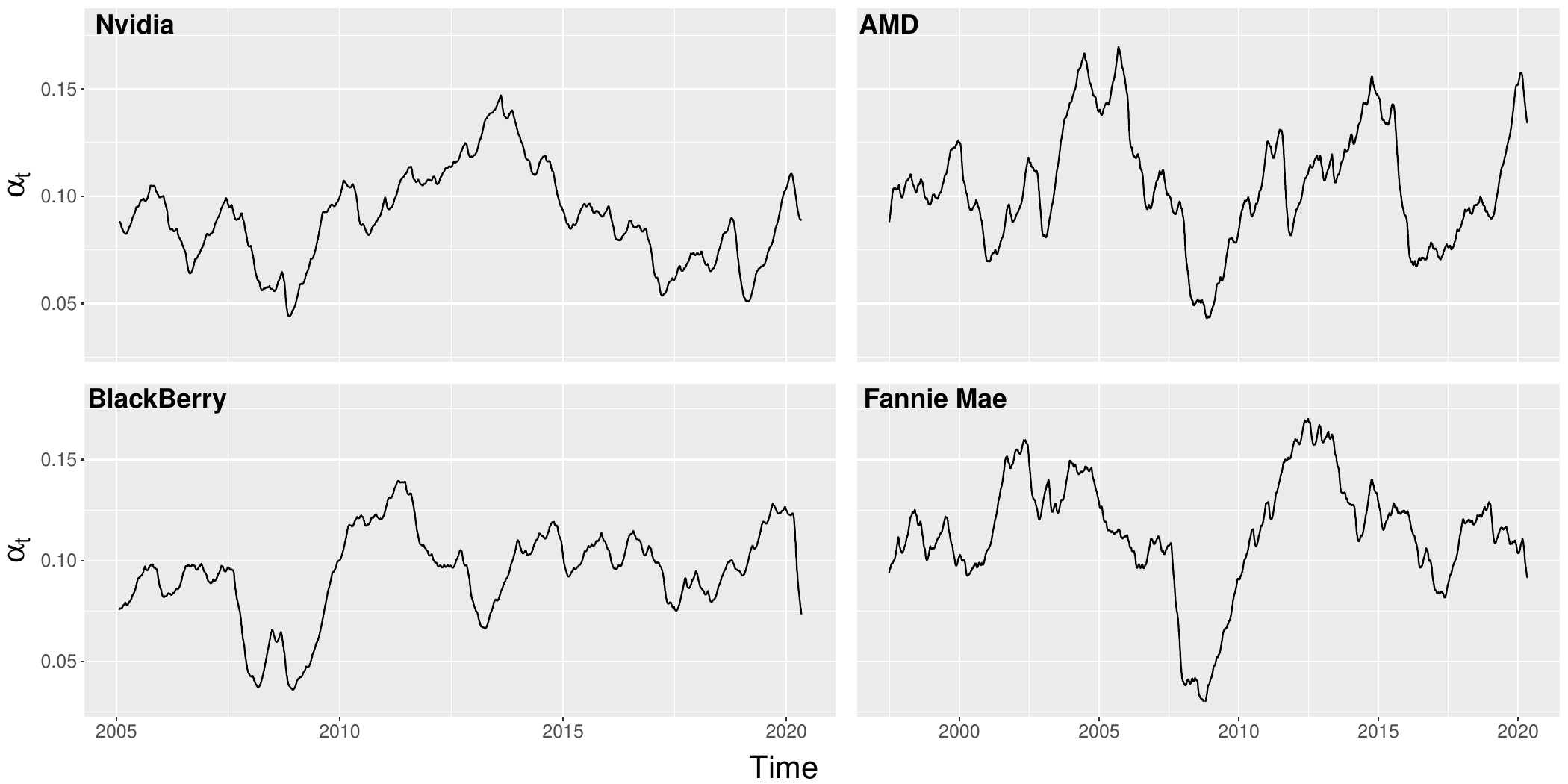}
  \caption{Realized trajectories of $\alpha_t$ for predicting stock market volatility as outlined in Section \ref{sec:stock_prediction} using update (\ref{eq:local_weighted_alpha_update}) with $w_s \propto 0.95^{t-s}$.}
\label{fig:stock_alpha_traj_momentum}
\end{figure}

 \begin{figure}[H]
  \centering
  \includegraphics[scale=0.36]{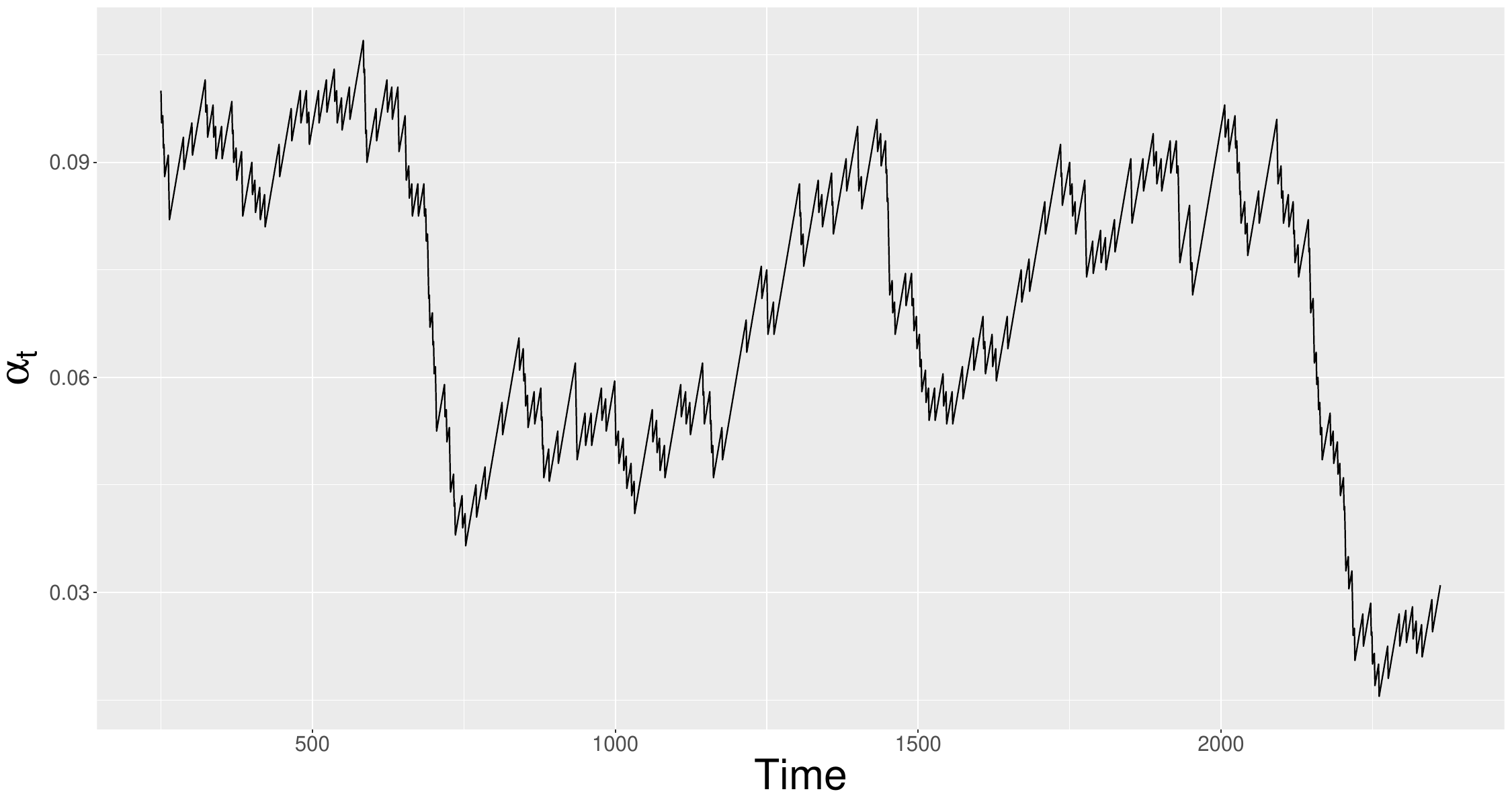}
  \caption{Realized trajectory of $\alpha_t$ for election night
    forecasting as outlined in Section \ref{sec:election_forecasting}
    using update (\ref{eq:simple_alpha_update}).}
  \label{fig:elec_alpha_traj_simple}
\end{figure}

\begin{figure}[H]
  \centering
  \includegraphics[scale=0.36]{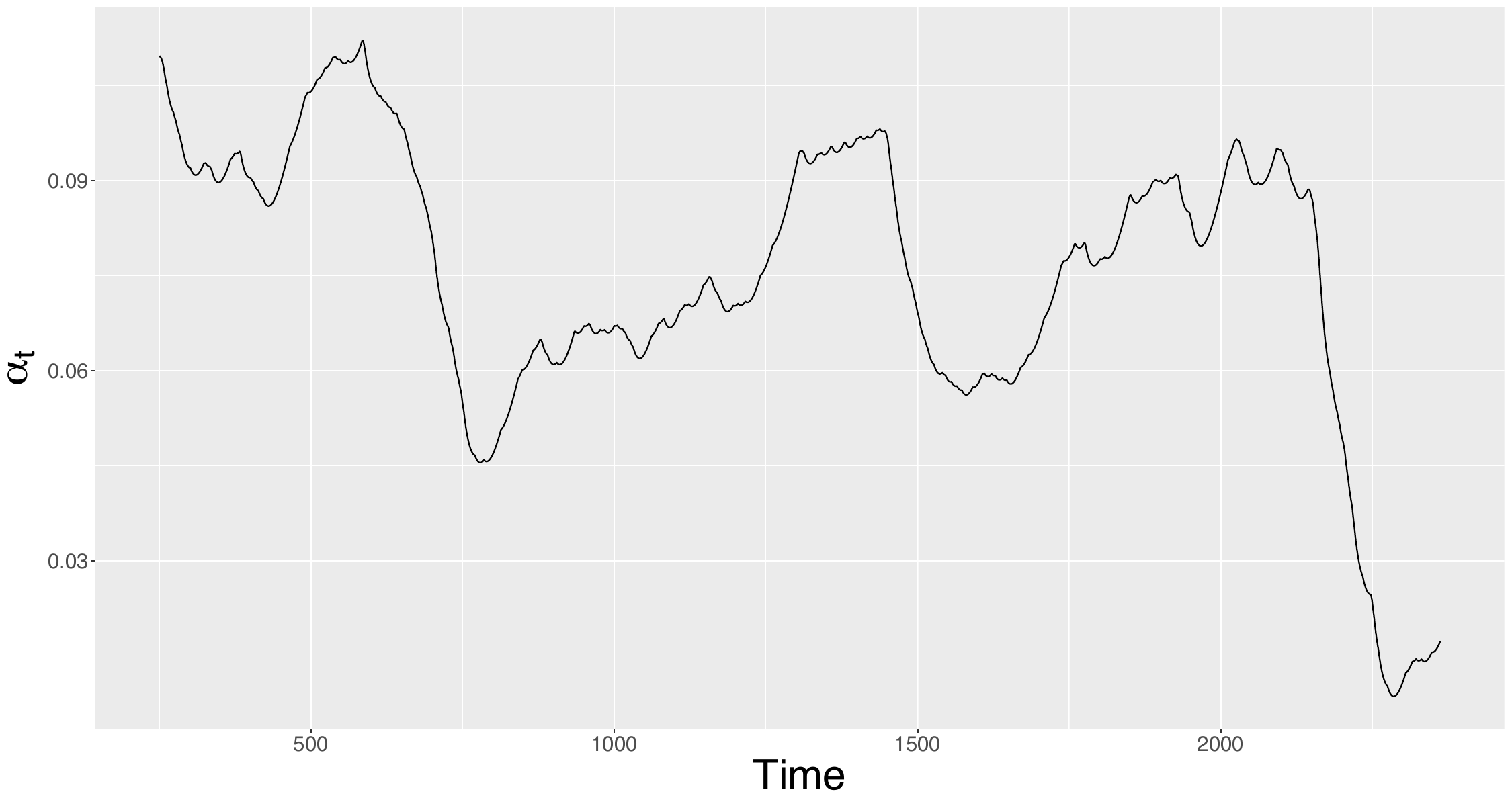}
  \caption{Realized trajectory of $\alpha_t$ for election night forecasting as outlined in Section \ref{sec:election_forecasting} using update (\ref{eq:local_weighted_alpha_update}) with $w_s \propto 0.95^{t-s}$. }
    \label{fig:elec_alpha_traj_momentum}
\end{figure}

\subsection{Coverage for additional stocks} 

Figure \ref{fig:additional_stocks} shows the local coverage level of adaptive and non-adaptive conformal inference for the prediction of market volatility (see Section \ref{sec:stock_prediction}) for 8 additional stocks/indices.

\begin{figure}[H]
  \centering
  \includegraphics[scale=0.36]{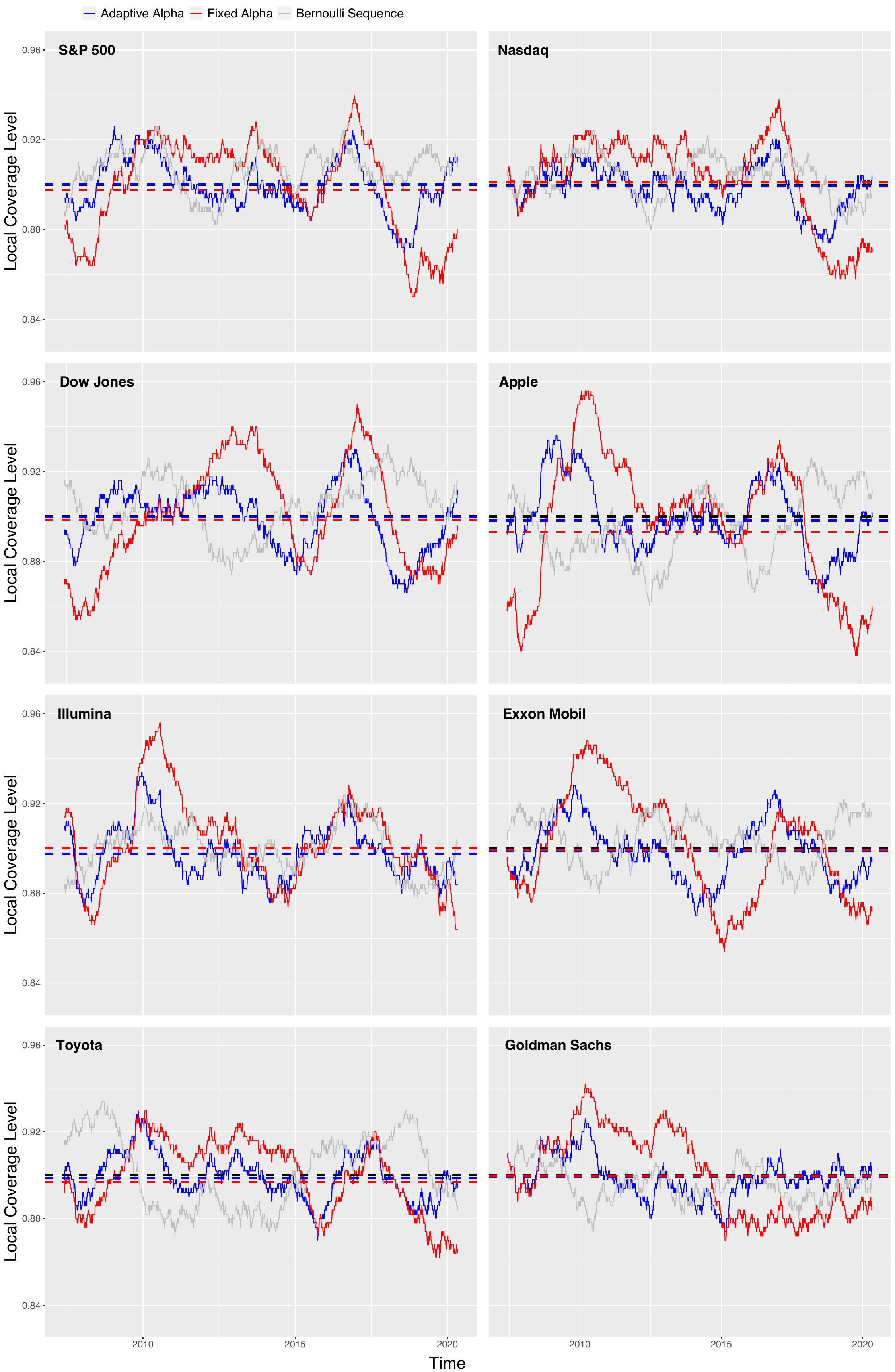}
 \caption{Local coverage frequencies for adaptive conformal (blue), a non-adaptive method that holds $\alpha_t = \alpha$ fixed (red), and an i.i.d. Bernoulli(0.1) sequence (grey) for the prediction of market volatility. The coloured dotted lines mark the average coverage obtained across all time points, while the black line indicates the target level of $1-\alpha=0.9$. }
    \label{fig:additional_stocks}
\end{figure}

\subsection{Existence of a stationary distribution for $\pmb{(\alpha_t,A_t)}$}\label{sec:MC_stationary_dist}

In this section we give one simple example in which the Markov chain $\{(\alpha_t,A_t)\}_{t \in \mmn}$ will have a unique stationary distribution. The setting considered here is the same as the one described in Section \ref{sec:MC_converage_bounds}.

Let $\alpha_t$ be initialized with $\alpha_1 \in \{\alpha + k\gamma \alpha : k \in \mmz\}$. Assume that $\mathcal{A}$ is finite. Let $P$ be the transition matrix of $\{A_t\}_{t \in \mmn}$ and assume that for all $a_1,a_2 \in \mathcal{A}$, $P_{a_1,a_2} = \mmp(A_{t+1} =a_2|A_t =a_1)  > 0$. Assume that $\alpha$ satisfies $\alpha^{-1}(1-\alpha) \in \mmn$. This will hold for common choices of $\alpha$ such as $\alpha = 0.1$ and $\alpha = 0.05$. Finally, assume that for all $a \in \mathcal{A}$ and all $p \in (0,1)$, $\mmp( S(X_t,Y_t) \leq \hat{Q}(p)|A = a) \in (0,1)$. This will occur for example when $S(X_t,Y_t)|A_t=a$ is supported on $\mathbb{R}$ and $\hat{Q}(\cdot)$ is finite valued for all $p \in (0,1)$. 

We claim that in this case $\{(\alpha_t,A_t)\}_{t \in \mmn}$ has a unique stationary distribution. To prove this it is sufficient to show that this chain is irreducible and has a finite state space. To check that it has a finite state space it is sufficient to show that $\alpha_t$ has a finite state space. We claim that with probability one we have that for all $t \in \mmn$, $\alpha_t \in \{\alpha + k\gamma \alpha : k \in \mmz\}$.  To prove this we proceed by induction. The base case is given by our choice of $\alpha_1$. For the inductive step note that
\begin{align*}
& \alpha_t \in \{\alpha + k\gamma \alpha : k \in \mmz\}\\
& \implies \alpha_{t+1} = \alpha_t + \gamma (\alpha - \text{err}_{t}) = \begin{cases}
\alpha_t + \gamma \alpha , \text{ if } \text{err}_t = 0,\\
\alpha_t - \gamma \alpha ( \alpha^{-1}(1-\alpha)), \text{ if } \text{err}_t = 1,
\end{cases}
 \in \{\alpha + k\gamma \alpha : k \in \mmz\} .
\end{align*}
Since $\alpha_t$ is also bounded (see Lemma \ref{lem:alphat_is_bounded}) this implies that $\alpha_t$ has a finite state space.

Finally, the fact that $\{(\alpha_t,A_t)\}_{t \in \mmn}$ is irreducible follows easily from our assumptions on $\{A_t\}_{t \in \mmn}$, $\hat{Q}(\cdot)$, and $\mmp( S(X_t,Y_t) \leq \hat{Q}(p)|A = a)$.

\subsection{Additional information for Section \ref{sec:election_forecasting}}

\subsubsection{Dataset description}\label{sec:election_dat}

The county-level demographic characteristics used for prediction were the proportion of the total population that fell into each of the following race categories (either alone or in combination): black or African American,  American Indian or Alaska Native, Asian, Native Hawaiian or other Pacific islander. In addition to this, we also used the proportion of the total population that was male, of Hispanic origin, and that fell within each of the age ranges 20-29, 30-44, 45-64, and 65+.  Demographic information was obtained from 2019 estimates published by the United States Census Bureau and available at \cite{censusDemos2019}. To supplement these demographic features we also used the median household income and the percentage of individuals with a bachelors degree or higher as covariates. Data on county-level median household incomes was based on 2019 estimates obtained from \cite{censusIncom}. The percentage of individuals with a bachelors degree or higher was computed based on data collected during the years 2015-2019 and published at \cite{censusEdu}. As an aside, we remark that we used 2019 estimates because this was the most recent year for which data was available. 

Vote counts for the 2016 election were obtained from \cite{electionsData}, while 2020 election data was taken from \cite{Leip2020}. In total, matching covariate and election vote count data were obtained for 3111 counties.

\subsubsection{Detailed prediction algorithm}\label{sec:cqr_alg}

Algorithm \ref{alg:wp_split_conformal_method} below outlines the core conformal inference method used to predict election results.  An R implementation of this algorithm as well as the core method outlined in Section \ref{sec:stock_prediction} can be found at \url{https://github.com/isgibbs/AdaptiveConformal}.

\begin{algorithm}\label{alg:wp_split_conformal_method}
 \KwData{Observed sequence of county-level votes counts and covariates $\{(X_t,Y_t)\}_{1 \leq t \leq T}$ and vote counts for the democratic candidate in the previous election $\{Y_t^{\text{prev}}\}_{1 \leq t \leq T}$.}
 \For{$t=1,2,\dots,T$}{
 	Compute the residual $r_t = (Y_t - Y_t^{\text{prev}})/Y_t^{\text{prev}}$
 }
 \For{$t = 501, 502,\dots,T$}{\tcp{We start making predictions once 500 counties have been observed.} 
 	Randomly split the data $\{(X_l,r_l)\}_{1 \leq l \leq t-1}$ into a training set $\mathcal{D}_{\text{train}}$ and a calibration set $\mathcal{D}_{\text{cal}}$ with $|\mathcal{D}_{\text{train}}| = \lfloor (t-1) \cdot 0.75 \rfloor$\;
 	{Fit a linear quantile regression model $\hat{q}(x;p)$ on $\mathcal{D}_{\text{train}}$}\;
 	\For{$(X_l,r_l) \in \mathcal{D}_{\text{cal}}$}{
 	 Compute the conformity score $S_l = \max\{\hat{q}(X_l;\alpha/2)-r_l,r_l-\hat{q}(X_l;1-\alpha/2)\}$\;
 	}
 	Define the quantile function $\hat{Q}_{t}(p) = \inf\left\{x : \left( \frac{1}{|\mathcal{D}_{\text{cal}}|} \sum_{(X_l,r_l) \in \mathcal{D}_{\text{cal}}} \bone_{S_l \leq x}  \right)  \geq p\right\}$\;
 	Return the prediction set $\hat{C}_{t}(\alpha) := \{y : \max\{\hat{q}(X_t;\alpha/2) - \frac{y-Y_t^{\text{prev}}}{Y_t^{\text{prev}} } ,  \frac{y-Y_t^{\text{prev}}}{Y_t^{\text{prev}}}  - \hat{q}(X_t;1-\alpha/2)\} \leq \hat{Q}_{t}(1-\alpha)\}$\;
 }
 \caption{\textit{CQR method for election night prediction}}
\end{algorithm}

\subsection{Large deviation bounds for the error sequence}\label{sec:large_deviation_proof}

In this section we prove Theorem \ref{thm:large_dev_err_bound}. So, throughout we define the sequences $\{\alpha_t\}_{t \in \mmn}$, $\{\text{err}_t\}_{t \in \mmn}$, and $\{A_t\}_{t \in \mmn}$ as in Section \ref{sec:MC_converage_bounds} and we assume that $\{(\alpha_t,A_t)\}_{t \in \mmn}$ is a stationary Markov chain from which it follows immediately that $\{(\alpha_t,A_t,\text{err}_t)\}_{t \in \mmn}$ is also stationary. Additionally, we assume that $\hat{Q}(\cdot)$ and $S(\cdot)$ are fixed functions such that $\hat{Q}(\cdot)$ is non-decreasing with $\hat{Q}(x) = - \infty$ for all $x<0$ and $\hat{Q}(x) = \infty$ for all $x>1$. The proof of Theorem \ref{thm:large_dev_err_bound} will rely on the following lemmas.

\begin{lemma}\label{lem:neg_corr_mon_fns}
Let $f:\mmr \to \mmr$ and $g:\mmr \to \mmr$ be bounded functions such that either 
\begin{enumerate}
\item 
$f$ is non-increasing and $g$ is non-decreasing,
\item
or $f$ is non-decreasing and $g$ is non-increasing.
\end{enumerate}
Then, for any random variable $Y$ 
\[
\mme[f(Y)g(Y)] \leq \mme[f(Y)]\mme[g(Y)].
\] 
\end{lemma} 
\noindent The proof of this result is straightforward and can be found in Section \ref{sec:technical_lemmas}.

\begin{lemma}\label{lem:removing_alphat_cor}
For any $\lambda \in \mmr$ and $t \in \mmn$ we have that 
\begin{equation}\label{eq:inductive_statement}
\mme\left[\prod_{s=1}^{t} \exp(\lambda(\textup{err}_s - \mme[\textup{err}_s|A_s]))\right] \leq  \exp(\lambda^2/2) \mme\left[\prod_{s=1}^{t-1} \exp(\lambda(\textup{err}_s - \mme[\textup{err}_s|A_s]))\right].
\end{equation}
\end{lemma}
\begin{proof}
By conditioning on $\alpha_1$ and $A_1,\dots,A_t$ on both the left and right-hand side of (\ref{eq:inductive_statement}) we may view these quantities as fixed. Thus, while for readability we do not denote this conditioning explicitly, the following calculations should be read as conditional on $\alpha_1$ and $A_1,\dots,A_t$. Then, 
\begin{align*}
& \mme\left[\prod_{s=1}^{t} \exp(\lambda(\text{err}_s - \mme[\text{err}_s|A_s]))\right]\\
& \ \ \ \ \ \  = \mme\left[\prod_{s=1}^{t-1} \exp(\lambda(\text{err}_s - \mme[\text{err}_s|A_s]))\mme\bigg[\exp(\lambda(\text{err}_t -  \mme[\text{err}_t|A_t])) \bigg| \text{err}_1,\dots,\text{err}_{t-1} \bigg] \right].
\end{align*}
Recall that $\alpha_t = \alpha_1 + \gamma\sum_{s=1}^{t-1}(\alpha-\text{err}_s)$ is a deterministic function of $\alpha_1$  and $\sum_{s=1}^{t-1} \text{err}_s$. So, we may define the functions $f(\sum_{s=1}^{t-1} \text{err}_s) = \prod_{s=1}^{t-1} \exp(\lambda(\text{err}_s - \mme[\text{err}_s|A_s])) $ and 
\begin{align*}
 g(\sum_{s=1}^{t-1} \text{err}_s) & := \mme\left[\exp(\lambda(\text{err}_{t} - \mme[\text{err}_{t}|A_t]))| \text{err}_1,\dots,\text{err}_{t-1}\right]\\
& =   P_{A_{t}}(S(X_t,Y_t) \leq \hat{Q}(1-\alpha_{t}))\exp(-\lambda \mme[\text{err}_t|A_t])\\
& \ \ \ \ + (1-P_{A_{t}}(S(X_t,Y_t) \leq \hat{Q}(1-\alpha_{t}))) \exp(\lambda(1- \mme[\text{err}_t|A_t])),
\end{align*}
where we emphasize that on the last line $A_t$ and $\alpha_t$ should be viewed as fixed quantities. Now, since $\alpha_{t}$ is monotonically decreasing in $\sum_{s=1}^{t-1} \text{err}_s$ it should be clear that if $\lambda \geq 0$ then $g$ is non-increasing (resp. non-decreasing for $\lambda < 0$) and $f$ is non-decreasing (resp. non-increasing for $\lambda < 0$). So, by Lemma \ref{lem:neg_corr_mon_fns} we have that 
\begin{align*}
& \mme\left[\prod_{s=1}^{t-1} \exp(\lambda(\text{err}_s - \mme[\text{err}_s|A_s]))\mme\bigg[\exp(\lambda(\text{err}_t -  \mme[\text{err}_t|A_t])) \bigg| \text{err}_1,\dots,\text{err}_{t-1} \bigg] \right] \\
& \leq \mme\left[\prod_{s=1}^{t-1} \exp(\lambda(\text{err}_s - \mme[\text{err}_s|A_s])) \right] \mme\left[\mme\bigg[\exp(\lambda(\text{err}_t -  \mme[\text{err}_t|A_t])) \bigg| \text{err}_1,\dots,\text{err}_{t-1} \bigg] \right] \\
& = \mme\left[\prod_{s=1}^{t-1} \exp(\lambda(\text{err}_s - \mme[\text{err}_s|A_s])) \right]  \mme[\exp(\lambda(\text{err}_t -  \mme[\text{err}_t|A_t]))]\\
& \leq \mme\left[\prod_{s=1}^{t-1} \exp(\lambda(\text{err}_s - \mme[\text{err}_s|A_s])) \right] \exp(\lambda^2/2),
\end{align*}
where the last inequality follows by Hoeffding's lemma (see Lemma \ref{lem:Hoef_lem}). 
\end{proof}

The final result we will need in order to prove Theorem \ref{thm:large_dev_err_bound} is a large deviation bound for Markov chains.

\begin{definition}\label{def:absolution_spectral_gap}
Let $\{X_i\}_{i \in \mmn} \subseteq \mathcal{X}$ be a Markov chain with transition kernel $P$ and stationary distribution $\pi$. Define the inner product space 
\[
L^2(\mathcal{X},\pi) = \left\{h : \int_{\mathcal{X}} h(x)^2\pi(dx) < \infty\right\}
\]
with inner product
\[
\langle h_1,h_2 \rangle_{\pi} = \int _{\mathcal{X}} h_1(x)h_2(x) \pi(dx).
\]
For any $h \in L^2(\mathcal{X},\pi)$, let
\[
L^2(\mathcal{X},\pi) \ni Ph := \int h(y) P(\cdot,dy).
\]
Then, we say that $\{X_i\}_{i \in \mmn}$ has non-zero absolute spectral gap $1-\eta$ if
\[
\eta := \sup\left\{ \sqrt{\langle Ph,Ph \rangle_{\pi}} : \langle h,h \rangle_{\pi} = 1,\ \int_{\mathcal{X}} h(x)\pi(dx) = 0\right\} < 1.
\]
\end{definition}

\begin{theorem}[Theorem 1 in \cite{Jiang2020}]\label{thm:MC_boud_with_var}
Let $\{X_i\}_{i  \in \mmn} \subseteq \mathcal{X}$ be a stationary Markov chain with invariant distribution $\pi$ and non-zero absolute spectral gap $1-\eta > 0$. Let $f_i:\mathcal{X} \to [-C,C]$ be a sequence of functions with $\pi(f_i) = 0$ and define $\sigma^2 := \sum_{i=1}^n \pi(f_i^2)/n$. Then, for $\forall \epsilon > 0$,
\[
\mmp_{\pi}\left(\frac{1}{n} \sum_{i=1}^n f_i(X_i) \geq \epsilon\right) \leq \exp\left( -  \frac{n(1-\eta)\epsilon^2/2}{(1+\eta)\sigma^2 + 5C\epsilon} \right).
\]
\end{theorem}

Finally, we are ready to prove Theorem \ref{thm:large_dev_err_bound}.

\begin{proof}[Proof of Theorem \ref{thm:large_dev_err_bound}.] Write
\begin{align}
& \mmp\left( \left| \frac{1}{T} \sum_{t=1}^T \text{err}_t - \alpha \right| >\epsilon \right)\\
& \leq \mmp\left( \left| \frac{1}{T} \sum_{t=1}^T \text{err}_t - \mme[\text{err}_t|A_t] \right| >\epsilon/2 \right) + \mmp\left( \left| \frac{1}{T} \sum_{t=1}^T \mme[\text{err}_t|A_t] - \alpha \right| >\epsilon/2 \right) \label{eq:prob_decomp}
\end{align}
By applying lemma \ref{lem:removing_alphat_cor} inductively we have that for all $\lambda > 0$,
\begin{align*}
\mmp\left( \frac{1}{T} \sum_{t=1}^T \text{err}_t - \mme[\text{err}_t|A_t]  >\epsilon/2 \right) & \leq \exp(-T\lambda \epsilon/2) \mme\left[ \prod_{t=1}^T \exp(\lambda( \text{err}_t - \mme[\text{err}_t|A_t] )) \right]\\
& \leq    \exp(-T\lambda \epsilon/2) \exp(T\lambda^2/2),
\end{align*}
with an identical bound on the left tail. Choosing $\lambda = \epsilon/2$ gives the bound
\[
\mmp\left( \left| \frac{1}{T} \sum_{t=1}^T \text{err}_t - \mme[\text{err}_t|A_t] \right| >\epsilon/2 \right)  \leq 2 \exp(-T\epsilon^2/8).
\]
On the other hand, the second term in (\ref{eq:prob_decomp}) can be bounded directly using Theorem \ref{thm:MC_boud_with_var}.
\end{proof}

\subsection{Approximate marginal coverage}\label{sec:approx_marginal_proof}

In this section we prove Theorem \ref{thm:reg_bound}.

\begin{proof}[Proof of Theorem \ref{thm:reg_bound}]
Our proof follows similar steps to those presented in previous works on stochastic gradient descent under distribution shift \cite{Cao2019, Zhu2016}. First, note that
\[
(\alpha_{t+1} - \alpha^*_{A_t})^2 = (\alpha_{t} - \alpha^*_{A_t})^2 + 2\gamma(\alpha-\text{err}_t)(\alpha_{t} - \alpha^*_{A_t}) + \gamma^2(\alpha-\text{err}_t)^2.
\]
Now recalling that $M(\alpha^*_{A_t}|A_t) = \alpha$, we find that 
\begin{align*}
  -\mme[(\alpha-\text{err}_t)(\alpha_{t} - \alpha^*_{A_t})]  & = \mme[\mme[(\text{err}_t - \alpha)(\alpha_{t} - \alpha^*_{A_t})|A_t,\alpha_t]]\\
 & = \mme[(M(\alpha_{t}|A_t) - M(\alpha^*_{A_t}|A_t))(\alpha_t - \alpha^*_{A_t})]\\
 &   \geq \frac{1}{L} \mme[(M(\alpha_t|A_t) - M(\alpha^*_{A_t}|A_t))^2]\\
 &= \frac{1}{L} \mme[(M(\alpha_t|A_t) - \alpha)^2].
\end{align*}
Thus it follows that
\begin{align*}
& 2\gamma L^{-1} \sum_{t=1}^T \mme[(M(\alpha_t|A_t) - \alpha)^2]  \leq \sum_{t=1}^T  \mme[(\alpha_{t} - \alpha^*_{A_t})^2 -  (\alpha_{t+1} - \alpha^*_{A_t})^2   + \gamma^2(\alpha-\text{err}_t)^2] \\
& \leq \sum_{t=1}^T\mme[(\alpha_{t+1} - \alpha^*_{A_{t+1}})^2 -  (\alpha_{t+1} - \alpha^*_{A_{t}})^2] + \mme[(\alpha_{1} - \alpha^*_{A_{1}})^2]   + \gamma^2T\\
& \leq  \sum_{t=1}^T\mme[2\alpha_{t+1}(\alpha^*_{A_{t}} -  \alpha^*_{A_{t+1}} ) ]  +  (\alpha^*_{A_{T+1}})^2  + \mme[(\alpha_{1} - \alpha^*_{A_{1}})^2]   + \gamma^2T\\
& \leq \sum_{t=1}^T 2(1+\gamma) \mme[|\alpha^*_{A_{t+1}} - \alpha^*_{A_t}|] +  (\alpha^*_{A_{T+1}})^2  + \mme[(\alpha_{1} - \alpha^*_{A_{1}})^2]  + \gamma^2T,
\end{align*}
where the last inequality follows from Lemma \ref{lem:alphat_is_bounded}. So, re-arranging we get the inequality
\begin{align*}
& \frac{1}{T} \sum_{t=1}^T \mme[(M(\alpha_t|A_t) - \alpha)^2]\\
& \leq \frac{L}{2T\gamma} \left(\sum_{t=1}^T 2(1+\gamma) \mme[|\alpha^*_{A_{t+1}} - \alpha^*_{A_t}|] +  (\alpha^*_{A_{T+1}})^2  + \mme[(\alpha_{1} - \alpha^*_{A_{1}})^2]  + \gamma^2T \right).
\end{align*}
Finally, since $\{(A_t,\alpha_t)\}_{t \in \mmn}$ is stationary we may let $T \to \infty$ to get that
\[
\mme[(M(\alpha_t|A_t) - \alpha)^2] \leq \frac{L(1+\gamma)}{\gamma} \mme[|\alpha^*_{A_{t+1}} - \alpha^*_{A_t}|] + \frac{L\gamma}{2}
\]
as claimed.
\end{proof}

\subsection{Bounds on $\pmb{B}$  and $\pmb{\sigma^2_B}$}\label{sec:bias_bound}

In this section we bound the constants $B$ and $\sigma^2_B$ appearing in the statement of Theorem
\ref{thm:large_dev_err_bound}. Let
\begin{align*}
& \epsilon_1 = \sup_{k \in \{0,1,2,\dots\}} \sup_{a \in \mathcal{A}} \mme[|\alpha^*_{A_{T-k}} - \alpha^*_{A_{T-k-1}}||A_T = a]\\
\text{ and } & \epsilon_2 = \sup_{k \in \{0,1,2,\dots\}} \sup_{a \in \mathcal{A}} \mme[(\alpha^*_{A_{T-k}} - \alpha^*_{A_{T-k-1}})^2|A_T = a].
\end{align*}
Then, our main result is Lemma \ref{lem:bound_on_mean_given_a} which shows that 
\begin{equation}\label{eq:bound_on_bias}
B \leq C(\gamma + \gamma^{-1}(\epsilon_1 + \epsilon_2)) \ \ \ \text { and } \ \ \ \sigma^2_B \leq B^2,
\end{equation}
where the constant $C$ depends on how close $M(\cdot|a)$ is to the ideal linear function $M(p|a) = p$. \\

Plugging (\ref{eq:bound_on_bias}) into Theorem \ref{thm:large_dev_err_bound} gives a concentration inequality for $|T^{-1}\sum_{t=1}^T \text{err}_t - \alpha|$. In particular, suppose we use an optimal stepsize of $\gamma \propto \sqrt{\epsilon_1}$. Then, combining (\ref{eq:bound_on_bias}) with Theorem \ref{thm:large_dev_err_bound} roughly tells us that
\begin{equation}\label{eq:thm_bound}
 \left|\frac{1}{T}\sum_{t=1}^T \text{err}_t - \alpha\right| \leq  O\left(\max\left\{\frac{1}{\sqrt{T}},\frac{\sqrt{\epsilon_1}}{\sqrt{T(1-\eta)}}, \frac{\sqrt{\epsilon_1}}{T(1-\eta)}\right\} \right) .
\end{equation}
As a comparison it may be instructive to note that the more naive bound given in Proposition \ref{prop:bound_on_errs} can be written as
\begin{equation}\label{eq:prop_bound}
 \left|\frac{1}{T}\sum_{t=1}^T \text{err}_t - \alpha\right| \leq O\left( \frac{1}{T\gamma}\right) =  O\left( \frac{1}{T\sqrt{\epsilon_1}}\right) .
\end{equation}
A sharp reader may notice that the naive bound given in (\ref{eq:prop_bound}) actually goes to 0 faster in $T$ than the HMM-based bound shown in (\ref{eq:thm_bound}). While this is true, the bound  (\ref{eq:prop_bound}) has the highly undesirable property of increasing in $1/\sqrt{\epsilon_1}$, i.e. the bound increases as the size of the distribution shift decreases. On the other hand, the HMM-based bound  has the more intuitive property of decreasing with the distribution shift. 

The only remaining issue is to determine the size of $\sqrt{\epsilon_1/(1-\eta)}$. To provide some insight into this quantity note that there are two main regimes in which we expect
$|T^{-1} \sum_{t=1}^T\text{err}_t - \alpha|$ to be small:
\begin{enumerate}
\item
Environments in which the state $A_t$ changes frequently, but $|\alpha^*_{A_{t+1}} -  \alpha^*_{A_{t}}|$ is always small. In this case it is reasonable to expect $1-\eta$ to not be too small and so we anticipate that (\ref{eq:thm_bound}) will give a reasonable bound. 
\item
Environments in which the state changes very infrequently. In this case $\{A_t\}_{t \in \mmn}$ will mix slowly and so we expect $1-\eta$ to be quite small. Additionally, we also have that $\alpha^*_{A_{t+1}} = \alpha^*_{A_t}$ a large proportion of the time and thus $\epsilon_1$ will also be small. As a result, it is not immediately clear what (\ref{eq:thm_bound}) tells us about $|T^{-1} \sum_{t=1}^T\text{err}_t - \alpha|$. Below we give one simple example that demonstrates that (\ref{eq:thm_bound}) can also be a reasonable bound in this instance.
\end{enumerate}

\begin{example}
Let $\{A_t\}_{t \in \mmn}$ be the Markov chain with states $\{1,\dots,n\}$ and transition matrix
\[
P = \left(p - \frac{1-p}{n-1}\right) I + \frac{1-p}{n-1} 11^T,
\]
where $p \in [0,1]$ is taken to be very close to 1. Let $\Delta := \max_{i \neq j} |\alpha^*_i - \alpha^*_j|$. Then, we have that $\epsilon_1,\epsilon_2 \leq \Delta(1-p)$. Moreover, note that this chain has spectral gap $1-\eta \cong 1-p$. Thus, (\ref{eq:thm_bound}) simplifies to
\begin{equation}\label{eq:simplified_upper_bound}
 \left|\frac{1}{T}\sum_{t=1}^T \text{err}_t - \alpha\right| \leq  O\left(\max\left\{\frac{1}{\sqrt{T}}, \frac{\sqrt{\Delta}}{T\sqrt{1-p}}\right\} \right) .
\end{equation}
In particular, for $T > \Delta/(1-p)$ we find that the error sequence concentrates at a rate of $O(1/\sqrt{T})$, which is consistent with the behaviour of an i.i.d. Bernoulli sequence. Finally, to understand this restriction on $T$ note that given a starting state $j \in \{1,\dots,n\}$ we expect $\alpha_t$ to contract towards $\alpha_{j}^*$ at a rate of $(1-\gamma)$ and therefore to have that
\[
\frac{1}{T} \sum_{t=1}^T |\alpha_t - \alpha^*_{j}| \propto \frac{1}{T} \sum_{t=1}^T (1-\gamma)^{t-1} |\alpha_1 - \alpha^*_{j}| \leq \frac{\Delta}{T\gamma} =   \frac{\sqrt{\Delta}}{T\sqrt{1-p}} ,
\]
where here we assumed that $\alpha_1 \in \{\alpha^*_1,\dots,\alpha^*_n\}$. Thus, the second term in the maximum in (\ref{eq:simplified_upper_bound}) can be seen as accounting for the rate of covergence of $\alpha_t$ to $\alpha^*_{j}$ during the time that the Markov chain is in  state $j$ and given that $\alpha_1$ starts at some $\alpha^*_i$, $1 \leq i \leq n$.

\end{example}

We now derive (\ref{eq:bound_on_bias}).

\begin{lemma}\label{lem:bound_on_alpha_approx_given_state}
Assume that $\exists 0<c<1/(2\gamma)$ such that for all $a \in \mathcal{A}$ and all $p \in [-\gamma,1+\gamma]$,
\[
|M(p|a) - M(\alpha^*_a|a)| \geq c|p-\alpha^*_a|.
\]
Then, for all $a \in \mathcal{A}$, $k \in \{0,1,2,\dots\}$, and $T \in \mmn$
\begin{align*}
& \mme[(\alpha_1 - \alpha^*_{A_1})^2|A_{1+k} = a]\\
 & \leq (1-2c\gamma)^{T-1} \mme[(\alpha_1 - \alpha^*_{A_1})^2|A_{T+k} = a] \\
&   +  \sum_{t=2}^{T} (1-2c\gamma)^{T-t} \bigg(\gamma^2 +  2(1+\gamma)\mme[|\alpha^*_{A_{t-1}} - \alpha^*_{A_t}||A_{T+k} = a] + \mme[(\alpha^*_{A_{t-1}} - \alpha^*_{A_t})^2|A_{T+k} = a]\bigg).
\end{align*}
Furthermore, if we assume that $\forall a \in \mathcal{A}$ and $k \in \{0,1,2,\dots\}$,
\[
\mme[|\alpha^*_{A_{t-k}} - \alpha^*_{A_{t-k-1}}||A_t = a] \leq \epsilon_1 \text{ and } \mme[(\alpha^*_{A_{t-k}} - \alpha^*_{A_{t-k-1}})^2|A_t = a] \leq \epsilon_2,
\]
then we find that 
\[
 \mme[(\alpha_1 - \alpha^*_{A_1})^2|A_{1+k} = a] \leq \frac{1}{2c\gamma} \left(\gamma^2 + 2(1+\gamma) \epsilon_1 + \epsilon_2 \right).
\]
\end{lemma}
\begin{proof}
Fix any $T \in \mmn$. Since $\{(\alpha_t,A_t)\}_{t \in \mmn}$ is stationary we have that 
\[
\mme[(\alpha_1 - \alpha^*_{A_1})^2|A_{1+k}] = \mme[(\alpha_T - \alpha^*_{A_T})^2|A_{T+k} ].
\]
Now note that
\begin{align*}
 \mme[(\alpha_T - \alpha^*_{A_T})^2|A_{T+k}]  & = \mme[(\alpha_T - \alpha^*_{A_{T-1}})^2|A_{T+k} ] + 2\mme[(\alpha_T - \alpha^*_{A_{T-1}})(\alpha^*_{A_{T-1}} - \alpha^*_{A_T})|A_{T+k} ]\\
 & \ \ \ \ +  \mme[(\alpha^*_{A_{T-1}} - \alpha^*_{A_T})^2|A_{T+k} ] \\
    & \leq  \mme[(\alpha_T - \alpha^*_{A_{T-1}})^2|A_{T+k}] + 2(1+\gamma)\mme[|\alpha^*_{A_{T-1}} - \alpha^*_{A_T}||A_{T+k} ]\\
    & \ \ \ \ + \mme[(\alpha^*_{A_{T-1}} - \alpha^*_{A_T})^2|A_{T+k} ],
\end{align*}
where on the last line we have applied Lemma \ref{lem:alphat_is_bounded}. The first term above can be bounded as
\begin{align*}
 & \mme[(\alpha_T - \alpha^*_{A_{T-1}})^2|A_{T+k}]\\
 & \leq  \mme[(\alpha_{T-1} - \alpha^*_{A_{T-1}})^2|A_{T+k} ] + 2\mme[\gamma(\alpha - \text{err}_{T-1})(\alpha_{T-1} - \alpha^*_{A_{T-1}})|A_{T+k}] + \gamma^2,
  \end{align*}
where we additionally have that
\begin{align*}
 &  \mme[(\alpha - \text{err}_T)(\alpha_{T-1} - \alpha^*_{A_{T-1}})|A_{T+k}]\\
 & = \mme[(\alpha - \mme[\text{err}_{T-1}|A_{T-1},\alpha_{T-1},A_{T+k}])(\alpha_{T-1} - \alpha^*_{A_{T-1}})|A_{T+k}]\\
 &  =   \mme[( M(\alpha^*_{A_{T-1}}|A_{T_1}) - M(\alpha_{T-1}|A_{T-1}))(\alpha_{T-1} - \alpha^*_{A_{T-1}})|A_{T+k}]\\
 & \leq - c\mme[(\alpha_{T-1} - \alpha^*_{A_{T-1}})^2|A_{T+k}].
  \end{align*}
Whence, 
  \[
 \mme[(\alpha_T - \alpha^*_{A_{T-1}})^2|A_{T+k} ] \leq (1-2c\gamma) \mme[(\alpha_{T-1} - \alpha^*_{A_{T-1}})^2|A_{T+k} ] + \gamma^2, 
 \]
 and plugging this into our first inequality yields
\begin{align*}
 \mme[(\alpha_T - \alpha^*_{A_T})^2|A_{T+k} ] & \leq (1-2c\gamma) \mme[(\alpha_{T-1} - \alpha^*_{A_{T-1}})^2|A_{T+k} ]  + \gamma^2\\
& \ \ \ \ + 2(1+\gamma)\mme[|\alpha^*_{A_{T-1}} - \alpha^*_{A_T}||A_{T+k} ] + \mme[(\alpha^*_{A_{T-1}} - \alpha^*_{A_T})^2|A_{T+k} ].
  \end{align*}
Repeating this argument inductively gives 
\begin{align*}
& \mme[(\alpha_T - \alpha^*_{A_T})^2|A_{T+k} = a]\\
& \leq (1-2c\gamma)^{T-1} \mme[(\alpha_1 - \alpha^*_{A_1})^2|A_{T+k} = a] \\
&   +  \sum_{t=2}^{T} (1-2c\gamma)^{T-t} \bigg(\gamma^2 +  2(1+\gamma)\mme[|\alpha^*_{A_{t-1}} - \alpha^*_{A_t}||A_{T+k} = a] + \mme[(\alpha^*_{A_{t-1}} - \alpha^*_{A_t})^2|A_{T+k} = a]\bigg).
\end{align*}
The final part of the lemma follows by sending $T \to \infty$.
\end{proof}

\begin{lemma}\label{lem:bound_on_mean_given_a}
Assume that for all $a \in \mathcal{A}$ and $p \in [-\gamma,1+\gamma]$, $M(\cdot|a)$ admits the second order Taylor expansion
\[
M(p|a) - M(\alpha^*_a|a) = C^1_a(p-\alpha^*_a) + C^2_{p,a}(p-\alpha^*_a)^2,
\]
where $0<c_1 \leq C^1_a \leq C_1 < 1/\gamma$ and $|C^2_{p,a}| \leq C_2$. Then, for all $T \in \mmn$ and $a \in \mathcal{A}$,
 \begin{align*}
\left| \mme[\textup{err}_1|A_1 = a] - \alpha \right| & \leq  C_1(1-\gamma c_1)^{T-1}\mme[|\alpha_1 - \alpha^*_{A_1}||A_T=a] \\
&   + \sum_{t=1}^{T-1} C_1 C_2 \gamma (1-c_1\gamma)^{T-t-1}\mme[(\alpha_t - \alpha_{A^*_t})^2|A_T=a]\\
& +   \sum_{t=1}^{T-1} C_1(1-c_1\gamma)^{T-t-1}\mme[|\alpha^*_{A_{t+1}} - \alpha^*_{A_t}||A_T]  + C_2 \mme[(\alpha_T - \alpha^*_{A_T})^2|A_T=a].
\end{align*}
Furthermore, suppose the assumptions of Lemma \ref{lem:bound_on_alpha_approx_given_state} hold and that $\forall a \in \mathcal{A}$ and $k \in \{0,1,2,\dots\}$,
\[
\mme[|\alpha^*_{A_{T-k}} - \alpha^*_{A_{T-k-1}}||A_T = a] \leq \epsilon_1 \text{ and } \mme[(\alpha^*_{A_{T-k}} - \alpha^*_{A_{T-k-1}})^2|A_T = a] \leq \epsilon_2.
\]
Then $\forall a \in \mathcal{A}$,
\[
\left| \mme[\textup{err}_1|A_1 = a] - \alpha \right| \leq \left(C_1C_2 \frac{1}{c_1} + C_2 \right)\frac{1}{2c\gamma} (\gamma^2 + 2(1+\gamma) \epsilon_1 + \epsilon_2)+ \frac{C_1}{c_1\gamma}\epsilon_1.
 \]
\end{lemma}
\begin{proof}
Fix any $T \in \mmn$. Since $\{(\text{err}_t,A_t)\}_{t \in \mmn}$ is stationary we have that 
\begin{align*}
\mme[\text{err}_1|A_1=a] = \mme[\text{err}_T|A_T=a].
\end{align*}
Then, by Taylor expanding $M(\cdot|A_T)$ we find that
\begin{align*}
 & \left|\mme[\text{err}_T|A_T=a] - \alpha \right| = \left|\mme[M(\alpha_T|A_T) - M(\alpha^*_{A_T}|A_T)|A_T]\right|\\
 & \leq  \left|\mme[C^1_{A_T}(\alpha_T - \alpha^*_{A_T})|A_T]\right| + \mme[C_{\alpha_T,A_T}^2(\alpha_T - \alpha^*_{A_T})^2|A_T]\\
 & \leq  \left|\mme[C^1_{A_T}(\alpha_T - \alpha^*_{A_T})|A_T]\right|  + C_2 \mme[(\alpha_T - \alpha^*_{A_T})^2|A_T].
 \end{align*}
 The first term above can be further bounded as
 \begin{align*}
& \left|\mme[C^1_{A_T}(\alpha_T - \alpha^*_{A_T})|A_T] \right|\\
&  =\left| \mme[C^1_{A_T}(\alpha_{T-1} + \gamma(\alpha - \text{err}_{T-1}) - \alpha^*_{A_{T-1}})|A_T] + \mme[C^1_{A_T}(\alpha_{A_{T-1}}^* - \alpha^*_{A_T})|A_T] \right| \\
 & \leq \left| \mme[C_{A_T}^1(\alpha_{T-1} - \alpha^*_{A_{T-1}})] + \gamma\mme[ C^1_{A_T}(M(\alpha^*_{A_{T-1}}|A_{T-1}) - M(\alpha_{T-1}|A_{T-1}))|A_T] \right|\\
 & \ \ \ \ \ \ \ \ \ \ \ \ \ \ \  + C_1\mme[|\alpha_{A_{T-1}}^* - \alpha^*_{A_T}||A_T]\\
  & =\left| \mme[C^1_{A_T}(1-\gamma C^1_{A_{T-1}})(\alpha_{T-1} - \alpha^*_{A_{T-1}})|A_T]\right|  +  C_1C_2\gamma\mme[ (\alpha_{T-1} - \alpha^*_{A_{T-1}})^2 | A_T]\\
  & \ \ \ \ \ \ \ \ \ \ \ \ \ \ \  + C_1\mme[|\alpha_{A_{T-1}}^* - \alpha^*_{A_T}||A_T].
\end{align*}
The desired result follows by repeating this process inductively. Finally, the last part of the Lemma follows by sending $T \to \infty$ and applying the result of Lemma \ref{lem:bound_on_alpha_approx_given_state}.

\end{proof}

As a final aside we remark that in the main text we claimed that in the ideal case where $M(p|a) = p$ for all $p \in [0,1]$ this bound can be replaced by 
\[
\left| \mme[\textup{err}_1|A_1 = a] - \alpha \right| \leq 2(\gamma + \gamma^{-1} \epsilon_1).
\]
This can be justified by using the fact that in this case we have that for all $p \in [-\gamma,1 + \gamma]$ 
\[
M(p|a) - M(\alpha^*_a|a) = (p-\alpha^*_a) + C^2_{p,a}
\]
with $|C^2_{p,a}| \leq \gamma$. The desired result then follows by repeating the argument of Lemma \ref{lem:bound_on_mean_given_a}.

\subsection{Technical lemmas}\label{sec:technical_lemmas}

\begin{proof}[Proof of Lemma \ref{lem:neg_corr_mon_fns}:]
We assume without loss of generality that $f$ is non-decreasing and $g$ is non-increasing as otherwise one can simply multiply both $f$ and $g$ by $-1$.

Let $g^U := \sup\{g(y) : f(y) > \mme[f(Y)]$ and $g^L := \inf\{g(y) : f(y) \leq \mme[f(Y)]$. By the monotonicity of $f$ and $g$ we clearly have that $g^L \geq g^U$. Therefore,
\begin{align*}
& \mme[f(Y)g(Y)] - \mme[f(Y)]\mme[g(Y)] = \mme[(f(Y) - \mme[f(Y)])g(Y)]\\
& \leq   \mme[(f(Y) - \mme[f(Y)])g^L\bone_{f(Y) \leq \mme[f(Y)]}] + \mme[(f(Y) - \mme[f(Y)])g^U\bone_{f(Y) > \mme[f(Y)]}] \\
& \leq   \mme[(f(Y) - \mme[f(Y)])g^L\bone_{f(Y) \leq \mme[f(Y)]}] + \mme[(f(Y) - \mme[f(Y)])g^L\bone_{f(Y) > \mme[f(Y)]}] \\
& = 0,
\end{align*}
as desired. 
\end{proof}

\begin{lemma}\label{lem:Hoef_lem}[Hoeffding's Lemma \cite{Hoeffding1963}]
Let $X$ be a mean 0 random variable such that $X \in [a,b]$ almost surely. Then, for all $\lambda \in \mmr$
\[
\mme[\exp(\lambda X)] \leq \exp\left(\lambda^2\frac{(b-a)^2}{8}\right).
\]
\end{lemma}

\end{document}